\documentclass[11pt]{article}
\usepackage[a4paper,margin=1in]{geometry}
\usepackage{amsmath,amssymb,amsthm,mathtools,aliascnt,todonotes}
\usepackage[dvipsnames]{xcolor}
\usepackage{enumitem,graphicx,adjustbox,booktabs,tocloft,array,longtable}
\usepackage{microtype}
\usepackage{tikz}
\usetikzlibrary{arrows.meta,positioning,calc,fit}
\usepackage{hyperref}
\usepackage{cleveref}
\hypersetup{colorlinks=true,linkcolor=MidnightBlue,citecolor=MidnightBlue,urlcolor=MidnightBlue,
 pdftitle={Counting Paths and Trees via Exterior Algebra},pdfauthor={Daniel Lokshtanov, Fahad Panolan, Saket Saurabh, Meirav Zehavi and Jie Xue}}
\definecolor{softblue}{RGB}{241,244,248}
\definecolor{softgreen}{RGB}{242,245,242}
\definecolor{softorange}{RGB}{248,246,242}
\definecolor{softgray}{RGB}{247,247,247}
\renewcommand{\cftsecfont}{\normalfont}
\renewcommand{\cftsecpagefont}{\normalfont}
\newtheorem{theorem}{Theorem}[section]
\newaliascnt{lemma}{theorem}
\newtheorem{lemma}[lemma]{Lemma}
\aliascntresetthe{lemma}
\newaliascnt{corollary}{theorem}

\aliascntresetthe{corollary}
\newaliascnt{proposition}{theorem}
\newtheorem{proposition}[proposition]{Proposition}
\aliascntresetthe{proposition}
\newaliascnt{observation}{theorem}

\aliascntresetthe{observation}
\theoremstyle{definition}
\newaliascnt{definition}{theorem}
\newtheorem{definition}[definition]{Definition}
\aliascntresetthe{definition}
\newaliascnt{remark}{theorem}
\newtheorem{remark}[remark]{Remark}
\aliascntresetthe{remark}
\crefname{theorem}{Theorem}{Theorems}
\Crefname{theorem}{Theorem}{Theorems}
\crefname{lemma}{Lemma}{Lemmas}
\Crefname{lemma}{Lemma}{Lemmas}
\crefname{corollary}{Corollary}{Corollaries}
\Crefname{corollary}{Corollary}{Corollaries}
\crefname{proposition}{Proposition}{Propositions}
\Crefname{proposition}{Proposition}{Propositions}
\crefname{observation}{Observation}{Observations}
\Crefname{observation}{Observation}{Observations}
\crefname{definition}{Definition}{Definitions}
\Crefname{definition}{Definition}{Definitions}
\crefname{remark}{Remark}{Remarks}
\Crefname{remark}{Remark}{Remarks}

\newcommand{\R}{\mathbb{R}}
\newcommand{\E}{\mathbb{E}}
\newcommand{\Var}{\operatorname{Var}}
\newcommand{\calP}{\mathcal{P}}

\newcommand{\OO}{\mathcal O}
\newcommand{\Mat}{\operatorname{Mat}}
\newcommand{\Emb}{\operatorname{Emb}}
\newcommand{\Aut}{\operatorname{Aut}}
\newcommand{\HH}{\mathsf H}
\newcommand{\poly}{\operatorname{poly}}

\newcommand{\tr}{\operatorname{tr}}

\newcommand{\norm}[1]{\left\lVert #1\right\rVert}

\AtBeginDocument{%
  \let\originaltableofcontents\tableofcontents
  \renewcommand{\tableofcontents}{\par\noindent
    \begin{minipage}{\linewidth}\originaltableofcontents\end{minipage}\par}}
    \title{Counting Paths and Trees via Exterior Algebra}

\author{%
Fahad Panolan\thanks{%
University of Leeds, UK.
Email: \texttt{F.Panolan@leeds.ac.uk}.}
\and
Saket Saurabh\thanks{%
The Institute of Mathematical Sciences, HBNI, Chennai, India;
and University of Bergen, Norway.
Email: \texttt{saket@imsc.res.in}.}
\and
Meirav Zehavi\thanks{%
Ben-Gurion University of the Negev, Israel.
Email: \texttt{meiravze@bgu.ac.il}.}
\and
Jie Xue\thanks{%
New York University Shanghai, China.
Email: \texttt{jiexue@nyu.edu}.}
}
\date{}
\begin{document}
\maketitle
\begin{abstract}
We give randomized approximation algorithms for counting k-paths and k-forests in a host graph.  Here $k$ denotes the number of pattern vertices,  $n$ and $m$ denote the numbers of
host vertices and edges or arcs, $\varepsilon$ is the relative
error, and $\delta$ is the failure probability.
Our main results are:
\begin{itemize}
    \item \textbf{Paths.}
    We approximate the number of directed paths on $k$ vertices
    in
    $2^k k^{\OO(1)}(n+m)\varepsilon^{-2}\log(2/\delta)$
    arithmetic operations.

\item \textbf{Trees and forests.}
For every fixed $\eta>0$, we approximate the number of
non-induced copies of a given forest on $k$ vertices in
$(2+\eta)^k n^{\OO_\eta(1)}
\varepsilon^{-2}\log(2/\delta)$ arithmetic operations.
We also obtain algorithms for patterns with supplied
tree or path decompositions of bounded width.


\end{itemize}
Our path algorithm resolves a conjecture of Koutis and
Williams~[CACM 2016] and answers an open question of
Lokshtanov, Saurabh, and Zehavi~[SODA 2021] by giving
a $2^k\poly(n,\varepsilon^{-1})$-time approximation scheme.


Our algorithms combine exterior algebra with random matrix
estimators, using the tensor-train moment bound of Rakhshan
and Rabusseau~[AISTATS 2020].
For forests, we use a small-component separator to evaluate
the estimator efficiently. 
\end{abstract}
\newpage

%

\section{Introduction}\label{sec:introduction}
Counting small subgraphs is used to study the structure of biological
and other networks. For example, network motif analysis compares the
number of occurrences of a pattern in a given network with its
frequency in an appropriate random network model~\cite{Milo824}.
Such comparisons require estimating how often a pattern occurs,
rather than merely determining whether it occurs.
This application has motivated algorithms for approximately counting
paths and other small patterns in large
graphs~\cite{DBLP:conf/ismb/AlonDHHS08}.

We study approximate counting of graph embeddings. An embedding of a
pattern graph $F$ into a host graph $G$ is an injective map
$f:V(F)\to V(G)$ such that $f(u)f(v)\in E(G)$ for every
$uv\in E(F)$. For directed graphs, edge directions must also be
preserved. We impose no condition on pairs that are nonadjacent in $F$;
thus the embeddings need not be induced. Given $F$ and $G$, our goal
is to estimate the number of such maps. Our main graph results concern
paths and specified trees and forests on $k$ vertices. Throughout,
$n$ and $m$ denote the numbers of host vertices and edges or arcs, respectively.

\paragraph{Why approximate counting?}
Counting solutions can be substantially harder than finding one.
Valiant's theorem on the permanent gives a classical example:
determining whether a bipartite graph has a perfect matching is
polynomial-time solvable, whereas counting its perfect matchings is
$\#\mathrm P$-complete~\cite{DBLP:journals/tcs/Valiant79}.
The same distinction occurs when the size of the desired solution is
the parameter. An algorithm is fixed-parameter tractable if its running
time is $f(k)n^{\OO(1)}$ for some computable function $f$, where $n$ is the input length and $k$ is the parameter (assumed to be small) associated with the input.
The {\sc $k$-Path} decision problem admits such algorithms, for example
by color coding~\cite{DBLP:journals/jacm/AlonYZ95}. In contrast,
Flum and Grohe~\cite{FlumG04} showed that counting $k$-vertex paths
exactly is $\#\mathrm W[1]$-hard. Thus an exact fixed-parameter
algorithm for this counting problem is unlikely under the usual
parameterized counting assumption. Curticapean's
survey~\cite{DBLP:conf/iwpec/Curticapean18} gives broader background
on parameterized counting complexity.

Although exact counting is unlikely to admit a fixed-parameter
algorithm, approximate counting may still be possible in
fixed-parameter time. If the true count is $N$, we seek an estimate $\widehat N$
satisfying $(1-\varepsilon)N\le\widehat N\le(1+\varepsilon)N$
with probability at least $1-\delta$, for given
$0<\varepsilon,\delta<1$. Arvind and Raman~\cite{DBLP:conf/isaac/ArvindR02}
gave fixed-parameter approximation schemes for counting paths and
other patterns of bounded treewidth. Once such a scheme exists, a
central question is how small its exponential dependence on the
pattern size can be.

\paragraph{The running-time question for paths and trees.}
Alon et al.~\cite{DBLP:conf/ismb/AlonDHHS08} used color coding
to approximately count paths on $k$ vertices in
$\OO((2e)^k(n+m)\varepsilon^{-2}\log(2/\delta))$ time,
with relative error $\varepsilon$ and failure probability at most
$\delta$.
Alon and Gutner developed deterministic approximate counters using
balanced families of colorings~\cite{DBLP:journals/talg/AlonG10,DBLP:conf/iwpec/AlonG09}.  
Their improved construction~\cite{DBLP:conf/iwpec/AlonG09}
gives running time $(2e)^{k+\OO(\log^3 k)}(n+m)\log n$
for relative error $\varepsilon$ satisfying
$\varepsilon^{-1}=k^{\OO(1)}$.
Brand, Dell, and Husfeldt~\cite{BDH18} subsequently obtained a
$4^k\poly(k)(n+m)\varepsilon^{-2}$ randomized bound using exterior
algebra. Lokshtanov, Saurabh, and Zehavi~\cite[Theorem 1.1]{LSZ21}
improved the path base to $2.619$ using representative counters, and
obtained the same base for specified trees and, more generally,
patterns of any fixed treewidth. These developments concern
approximation; they do not remove the hardness of exact counting.

Koutis and Williams~\cite{DBLP:journals/cacm/KoutisW16} conjectured
that paths could be approximately counted in
$2^k\poly(n,\varepsilon^{-1})$ time.  
\footnote{$\poly(n,\varepsilon^{-1})
= n^{\OO(1)}\varepsilon^{-\OO(1)}$,
where the hidden constants are independent of $k$, $n$, and
$\varepsilon$.}
This gives a specific target for
the exponential dependence on $k$. Our path algorithm attains base
$2$, with a factor linear in $n+m$ apart from polynomial factors in
$k$ and the accuracy and confidence parameters. For specified trees
and forests, we obtain base $2+\eta$ for every fixed $\eta>0$.
Here the polynomial exponent in the host graph size $n+m$ depends on $\eta$. 
Bj\"orklund, Lokshtanov, Saurabh, and Zehavi~\cite{BLSZ21}
gave a deterministic approximate path counter that runs in
$4^{k+o(k)}(n+m)\log n$ time and polynomial space
for any fixed relative error. 
Our algorithms are randomized and use exponential space. 

\subsection{Our results}
We give precise statements below. Throughout, our approximation
guarantees use two parameters:
\begin{enumerate}
    \item $\varepsilon\in(0,1)$ is the relative error.
    If the true count  is $Z$, an estimate
    $\widehat{Z}$ is accurate when
    $(1-\varepsilon)Z\leq \widehat{Z}\leq(1+\varepsilon)Z$.

    \item $\delta\in(0,1)$ is the failure probability.
    The algorithm returns an accurate estimate with probability
    at least $1-\delta$, over its internal randomness.
\end{enumerate}
Thus, our guarantee is
$\Pr[|\widehat{Z}-Z|\leq\varepsilon Z]\geq 1-\delta$.
The running times below include the factor
$\varepsilon^{-2}\log(2/\delta)$: the dependence on the relative
error is quadratic in $1/\varepsilon$, and the dependence on the
failure probability is logarithmic in $1/\delta$.


\paragraph{Paths and forests.}
We first state our results for unweighted counting.

\begin{theorem}[{\sc Counting paths}]\label{thm:intro-path}
Let $G$ be a simple directed graph with $n$ vertices and $m$ arcs,
and let $0<\varepsilon,\delta<1$.
There is a randomized algorithm that approximates the number of
directed paths on $k$ vertices in $G$ within relative error
$\varepsilon$, with probability at least $1-\delta$, using
\begin{equation}\label{eq:intro-path-time}
  2^k k^{\OO(1)}(n+m)\varepsilon^{-2}\log(2/\delta)
\end{equation}
arithmetic operations.
\end{theorem}

\begin{theorem}[{\sc Counting forests}]\label{thm:intro-tree}
For every fixed $\eta>0$, there is a randomized algorithm
with the following guarantee.
Given a forest $F$ on $k$ vertices, a simple graph $G$ on $n$
vertices, and $0<\varepsilon,\delta<1$, the algorithm approximates
the number of embeddings of $F$ into $G$ within relative error
$\varepsilon$, with probability at least $1-\delta$, using
\begin{equation}\label{eq:intro-tree-time}
  (2+\eta)^k n^{\OO_\eta(1)}
  \varepsilon^{-2}\log(2/\delta)
\end{equation}
arithmetic operations.
The algorithm also applies when $F$ is a forest with each edge
assigned a direction and $G$ is directed.
For $0<\eta\leq 1$, the exponent of $n$ can be bounded by
$\OO(\eta^{-1}\log(2/\eta))$.
\end{theorem}

For forests, the polynomial exponent depends on the fixed
constant $\eta>0$. Thus, the bound approaches base $2$ as
$\eta$ decreases, at the cost of a larger polynomial exponent.

\subsection{Our approach}
We first describe the construction for counting $k$ vertex paths.
Assign an independent Gaussian vector in $\R^k$ to each
host vertex. The exterior product along a walk is zero
whenever a vertex is repeated. For a path on $k$ distinct
vertices, the product is a scalar multiple of the unique
degree-$k$ basis element, and this scalar is the determinant
of the assigned vectors. Thus exterior multiplication
removes walks with repeated vertices. However, the
determinants of different paths may have opposite signs,
so their sum does not give the number of paths.

To recover the count, we assign an independent random matrix
$M_{i,v}$ to each path position $i$ and host vertex $v$,
independently of the vertex vectors.
For a path $(v_1,\ldots,v_k)$, we obtain a scalar signature
from the ordered product
$M_{1,v_1}\cdots M_{k,v_k}$.
We multiply each path's determinant by its signature
and sum these contributions.
When we square this sum, the cross terms between distinct
paths have expectation zero, even if the paths share vertices.
The squared contribution of each individual path has the
same positive expectation. Dividing by this common value
therefore gives an unbiased estimator of the number of paths.

An unbiased estimator is useful only if its variance
is small enough that averaging over a small number of trials
gives an accurate estimate.  Conditional on the vertex vectors,
we apply the tensor-train moment bound of Rakhshan and
Rabusseau~\cite[Theorem 1 and Section 5.1]{RR20}
to the sum with determinant coefficients.
Combining this bound with Gaussian determinant moments
shows that one trial $X$ satisfies
\[
  \frac{\E[X^2]}{\E[X]^2}=\OO(k^2)
\]
whenever the count is positive.
Averaging independent trials and taking a median of these
averages gives relative error $\varepsilon$ with failure
probability at most $\delta$ using
$\OO(k^2\varepsilon^{-2}\log(2/\delta))$ trials.

\paragraph{Evaluating the estimator for paths and forests.}
For paths, a dynamic program sums the contributions of all
walks, while exterior multiplication removes those with
repeated vertices. Each transition appends one vertex and
multiplies the current exterior element by its vector.
There are $\binom{k}{d}$ exterior coordinates in degree $d$,
and
\[
  \sum_{d=0}^{k}\binom{k}{d}=2^k.
\]
Updating these coordinates and the auxiliary matrices takes
only polynomial additional work, giving the $2^k$ dependence.

For forests, combining partial embeddings requires multiplying
two exterior elements, rather than appending a single vector.
We use the small-component separator of Fomin, Lokshtanov,
Panolan, and Saurabh~\cite[Section 5.5.2, Lemma 5.18]{FLPS16}.
We enumerate the images of the separator vertices and process
the remaining components separately. The components are small
enough that each merge has a factor of low exterior degree,
which limits the number of coordinate pairs that must be
multiplied. Choosing the component-size threshold in terms
of $\eta$ gives the $(2+\eta)^k$ dependence; enumerating
separator images accounts for the $\eta$-dependent polynomial
exponent in the host size.
All separator assignments use the same randomness, and we
add their contributions before squaring.

The circuit bound improves the exponential base and the
$\varepsilon$ dependence of the bound
$\OO((3.841^k+s^{o(1)})\varepsilon^{-6}s)$
in~\cite[Theorem 1.2]{LSZ21}, while its displayed polynomial dependence
on normalized circuit size is larger. Theorem 1.3 of that work also
gives a $2.619^k$ base for skewed circuits.
Our coefficient-sum objective agrees with counting distinct multilinear
monomials when their coefficients are zero or one.

For weighted graph problems, summing products of weights differs from
minimizing an additive cost. Nederlof's dynamic representative
sets~\cite{Nederlof25} give an almost-$2^k$ deterministic algorithm
for minimum-weight directed paths; our algorithms estimate a total
weight. Related random projections include Gaussian tensor-network
embeddings~\cite{MS22} and Rademacher tensor-train
projections~\cite{RR22}. Our finite implementations and the
signed-permutation circuit construction are proved directly for the
counting objectives stated here.

%

\section{Preliminaries}
\label{sec:preliminaries}

For a positive integer $r$, write $[r]=\{1,\ldots,r\}$.
For a finite set $U$, let $\binom{U}{p}$ denote the family of its
$p$-element subsets. We use $S_p$ for the set of permutations of $[p]$
and $\operatorname{sgn}(\pi)\in\{-1,1\}$ for the sign of a
permutation $\pi\in S_p$. Unless a base is indicated, logarithms in
running-time bounds may be taken to base two.

\subsection{Graphs, paths, and embeddings}
All graphs are finite and simple. An undirected graph has no loops or
parallel edges. A directed graph has no loops or parallel arcs, but it
may contain both $(u,v)$ and $(v,u)$. For a host graph $G$, we write
$n=|V(G)|$ and $m=|E(G)|$; in a directed graph, $m$ counts arcs.

A walk on $p$ vertices is a sequence $(v_1,\ldots,v_p)$ in which each
consecutive pair is an edge. In a directed graph we require
$(v_i,v_{i+1})\in E(G)$ for every $i\in[p-1]$.
A walk is a path if its vertices are pairwise distinct.
Thus a path on $k$ vertices has $k-1$ edges or arcs.
We count directed paths as ordered vertex sequences.
A forest is an undirected graph with no cycles; it may be disconnected
and may have isolated vertices. A tree is a nonempty connected forest.
An orientation of a forest assigns one direction to each of its edges.
The underlying undirected graph is still a forest.

\begin{definition}[Embedding]
\label{def:prelim-embedding}
Let $F$ be a pattern graph and $G$ a host graph. An embedding of $F$
into $G$ is an injective map $f:V(F)\to V(G)$ that preserves edges:
\[
  ij\in E(F)\quad\Longrightarrow\quad f(i)f(j)\in E(G).
\]
For directed graphs, the map must preserve the direction of every arc.
We write $\Emb(F,G)$ for the set of these embeddings.
\end{definition}

Our embeddings are non-induced: no condition is imposed on the images
of nonadjacent pattern vertices. In particular, extra host edges between
image vertices are allowed. The pattern is specified as part of the
input. Counting embeddings means counting the maps in $\Emb(F,G)$,
with the pattern vertices distinguished. This may differ from counting
subgraphs isomorphic to $F$. For example, a pattern consisting of one
undirected edge has two embeddings onto each host edge, corresponding
to its two endpoint assignments.

We use $k=|V(F)|$ for the pattern size. If $k>n$, there are no
embeddings. The pattern with no vertices has one embedding, the empty
map. 

\subsection{Exterior algebra}
We define the exterior algebra in coordinates over $\R$, the field used
in our analysis. Fix a dimension $d$ and let $e_1,\ldots,e_d$ be the
standard basis of $\R^d$. The dimension $d$ will be $k$ in our algorithms.

For each subset $S\subseteq[d]$, introduce a basis symbol $e_S$.
The exterior algebra $\Lambda(\R^d)$ is the vector space whose elements
are formal linear combinations
\[
  x=\sum_{S\subseteq[d]}x[S]e_S,
  \qquad x[S]\in\R.
\]
Here $x[S]$ is the coordinate of $x$ indexed by $S$.
We identify $e_{\{i\}}$ with $e_i$ and write $e_\varnothing=1$.
Thus the original vector space $\R^d$ is the subspace spanned by the
singleton basis elements. In general, an exterior element has
coordinates indexed by subsets, rather than just singleton indices.

\begin{definition}[Exterior product]
\label{def:prelim-exterior-product}
For $A,B\subseteq[d]$, let
\[
  \iota(A,B)=|\{(a,b)\in A\times B:a>b\}|.
\]
Define the exterior product, also called the wedge product, on basis
elements by
\[
  e_A\wedge e_B=
  \begin{cases}
    (-1)^{\iota(A,B)}e_{A\cup B},& A\cap B=\varnothing,\\
    0,& A\cap B\ne\varnothing,
  \end{cases}
\]
and extend it bilinearly to all exterior elements. Explicitly,
\[
  x\wedge y=
  \sum_{A,B\subseteq[d]}x[A]y[B](e_A\wedge e_B).
\]
\end{definition}

For disjoint sets, the sign records the number of swaps needed to put
the concatenation of their increasing orders into increasing order.
For example,
\[
  e_{\{1,3\}}\wedge e_2=-e_{\{1,2,3\}},
  \qquad
  e_{\{1,3\}}\wedge e_3=0.
\]
The element $1=e_\varnothing$ is a multiplicative identity.

\paragraph{Associativity.}
The exterior product is associative:
\[
  (x\wedge y)\wedge z=x\wedge(y\wedge z).
\]
To check this, bilinearity reduces the claim to three basis elements
$e_A,e_B,e_C$. If any two of $A,B,C$ intersect, both products are zero.
Otherwise, the sign exponent in either product is
\[
  \iota(A,B)+\iota(A,C)+\iota(B,C).
\]
We may therefore write a product of several exterior elements without
specifying parentheses. The order of the factors still matters.

\paragraph{Degree and dimension.}
The $p$-th exterior power, denoted by $\Lambda^p(\R^d)$, is the
subspace spanned by the $e_S$ with $|S|=p$.
An element of this subspace is homogeneous of exterior degree $p$.
In particular,
\[
  \Lambda^0(\R^d)=\R,
  \qquad \Lambda^1(\R^d)=\R^d,
  \qquad \dim\Lambda^p(\R^d)=\binom{d}{p}.
\]
Grouping the basis elements by their subset sizes gives
\[
  \Lambda(\R^d)=\bigoplus_{p=0}^d\Lambda^p(\R^d),
  \qquad
  \dim\Lambda(\R^d)=\sum_{p=0}^d\binom{d}{p}=2^d.
\]
The product of elements of degrees $p$ and $q$ lies in degree $p+q$;
it is zero if $p+q>d$.

\paragraph{Order and repeated vectors.}
For $x\in\Lambda^p(\R^d)$ and $y\in\Lambda^q(\R^d)$, the basis
multiplication rule gives
\[
  x\wedge y=(-1)^{pq}y\wedge x.
\]
This is graded commutativity. For ordinary vectors $u,v\in\R^d$,
which have degree one, it gives
\[
  u\wedge v=-v\wedge u,
  \qquad v\wedge v=0.
\]
The second equality is the alternating property. If a longer product
of vectors contains the same vector twice, swap adjacent vectors until
the two copies are next to each other. Each swap changes only the sign,
and the adjacent repeated pair has product zero. Thus repeated vectors
make the entire product zero, regardless of their positions.

More generally, for vectors $v_1,\ldots,v_p\in\R^d$ and
$\pi\in S_p$,
\[
  v_{\pi(1)}\wedge\cdots\wedge v_{\pi(p)}
  =\operatorname{sgn}(\pi)\,
    v_1\wedge\cdots\wedge v_p.
\]
The product is multilinear in its vector arguments. For example,
\[
  (au+bv)\wedge w
  =a(u\wedge w)+b(v\wedge w)
  \qquad(a,b\in\R).
\]
An arbitrary element of $\Lambda^p(\R^d)$ is a linear combination
of products of $p$ vectors; it need not itself be a single such product.

\paragraph{The determinant coordinate.}
There is only one basis element in exterior degree $d$, namely
$e_{[d]}$. For $d$ vectors $v_1,\ldots,v_d\in\R^d$,
\[
  v_1\wedge\cdots\wedge v_d
  =\det[v_1\ \cdots\ v_d]e_{[d]}.
\]
Indeed, expanding by multilinearity leaves precisely the signed
permutation terms in the determinant. In lower degrees, the coordinates
are the corresponding minors; Lemma~\ref{lem:wedge-coordinate-prelim}
gives that formula and its proof.
Consequently, the product $v_1\wedge\cdots\wedge v_p$ is zero if
and only if the list of vectors $v_1,\ldots,v_p$ is linearly dependent.

\subsection{Matrix and probability notation}
For a real column vector $x$, write $\norm{x}$ for its Euclidean norm.
For a real matrix $A$, let $A^\top$ denote its transpose and let
\[
  \norm{A}_F^2=\sum_{i,j}A_{ij}^2
\]
be its squared Frobenius norm. For a square matrix $A$, write
$\tr(A)=\sum_i A_{ii}$ for its trace and $I_D$ for the identity matrix
of dimension $D$. The auxiliary matrix dimension $D$ is distinct from
the exterior dimension $k$; its value will be specified for each
construction.

We write $\E[X]$ for expectation and
$\Var(X)=\E[(X-\E[X])^2]$ for variance.
A standard Gaussian scalar has distribution $N(0,1)$, with mean zero
and variance one. A standard Gaussian vector has independent standard
Gaussian coordinates. We specify the entry variances and independence
conditions whenever Gaussian matrices are sampled.
The Gaussian moment identities needed in the analysis are proved in
Section~\ref{sec:moments}.

An estimator $Y$ for a count or total weight $W$ is unbiased if
$\E[Y]=W$. When $W>0$, its relative second moment is
\[
  \frac{\E[Y^2]}{W^2}
  =1+\frac{\Var(Y)}{W^2}.
\]
Thus a bound on the relative second moment also bounds the variance
relative to the square of the quantity being estimated.
For $0<\varepsilon,\delta<1$, a relative approximation must satisfy
\[
  \Pr\bigl[|\widehat W-W|\le\varepsilon W\bigr]\ge1-\delta.
\]
Here $\varepsilon$ controls the relative error and $\delta$ the
failure probability. We handle $W=0$ separately, without forming a
relative second-moment ratio.
\section{From walks to an exterior-algebra sieve}\label{sec:path-motivation}
We begin with unweighted directed paths. The purpose of this first
construction is to separate two questions: how to remove walks that
repeat a vertex, and how to count the paths that remain. Exterior algebra
answers the first question. The random estimator will answer the second. 

Let $G=(V,E)$ be a simple directed graph with $n$ vertices and $m$ arcs.
A $p$-vertex walk is a sequence $(v_1,\ldots,v_p)$ with
$(v_i,v_{i+1})\in E$ for each $i<p$. It is a path when its vertices are
distinct. Write $\calP_k$ for the ordered $k$-vertex paths and
$N_k=|\calP_k|$. Different vertex orders count as different paths.
We return zero if $k>n$, and return $n$ if $k=1$. In the construction
below, $2\le k\le n$.

\subsection{The algebraic rule that enforces distinctness}
Give every host vertex $v$ a vector $g_v\in\R^k$. For the moment these
vectors are arbitrary. We want a product that becomes zero whenever a
walk uses the same vector twice, even when the two occurrences are far
apart. The exterior product has exactly this property.

Let $e_1,\ldots,e_k$ be the standard basis of $\R^k$. The exterior algebra
$\Lambda(\R^k)$ is generated by these vectors, with a bilinear product
satisfying $e_i\wedge e_i=0$ and $e_i\wedge e_j=-e_j\wedge e_i$.
For $S=\{s_1<\cdots<s_p\}\subseteq[k]$, put
$e_S=e_{s_1}\wedge\cdots\wedge e_{s_p}$ and $e_\varnothing=1$.
The $e_S$ with $|S|=p$ form the basis of degree $p$, whose dimension is
$\binom{k}{p}$. The whole algebra has dimension $2^k$.


In particular, $g\wedge g=0$ for every vector $g$. If a longer product
contains $g$ twice, moving one occurrence next to the other changes only
its sign; the resulting product is zero. Thus
$g_a\wedge g_b\wedge g_a=0$, just as $g_a\wedge g_a=0$.
This is an exact identity, not a probabilistic test.

For later coordinate calculations, define
$\iota(A,B)=|\{(a,b)\in A\times B:a>b\}|$. The basis multiplication is
\begin{equation}\label{eq:wedge-product}
 e_A\wedge e_B=
 \begin{cases}
 (-1)^{\iota(A,B)}e_{A\cup B},&A\cap B=\varnothing,\\
 0,&A\cap B\ne\varnothing.
 \end{cases}
\end{equation}
The exponent counts the swaps needed to put the concatenated coordinates
in increasing order. For example, $e_{\{1,3\}}\wedge e_2=-e_{\{1,2,3\}}$.

\begin{lemma}[Coordinate formula for a wedge product]\label{lem:wedge-coordinate-prelim}
Let $a_1,\ldots,a_p\in\R^k$.  For every $p$-set $S=\{i_1<\cdots<i_p\}\subseteq[k]$, let $A_S$ be the $p\times p$ matrix whose $j$-th column is the restriction of $a_j$ to the coordinates in $S$.  Then
\[
  a_1\wedge\cdots\wedge a_p
  =
  \sum_{S\in\binom{[k]}{p}} \det(A_S)e_S.
\]
In particular, for $p=k$,
\[
  a_1\wedge\cdots\wedge a_k
  =
  \det(a_1,\ldots,a_k)e_{[k]}.
\]
\end{lemma}

\begin{proof}
Write $a_j=\sum_i a_j(i)e_i$.  By multilinearity,
\[
  a_1\wedge\cdots\wedge a_p
  =
  \sum_{i_1,\ldots,i_p} a_1(i_1)\cdots a_p(i_p)
  e_{i_1}\wedge\cdots\wedge e_{i_p}.
\]
All terms with repeated indices vanish.  For a fixed set $S=\{s_1<\cdots<s_p\}$, the surviving terms using exactly the indices of $S$ are indexed by permutations $\pi\in S_p$ and contribute
\[
  \sum_{\pi\in S_p}\operatorname{sgn}(\pi)
  \prod_{j=1}^p a_j(s_{\pi(j)})\, e_S
  =
  \det(A_S)e_S.
\]
Summing over all $S$ proves the formula.  The case $p=k$ has only the basis vector $e_{[k]}$.
\end{proof}

This formula explains what the top exterior coordinate means. There is
only one degree-$k$ basis vector, $e_{[k]}$. A path contributes its
determinant as the coefficient of this vector. Intermediate coordinates
are minors of the corresponding vector matrix.

\subsection{A dynamic program that sums over walks}
\label{subsec:warmup-exterior-detection}
Let $B_{p,v}\in\Lambda^p(\R^k)$ be the state for walks with $p$ vertices
ending at $v$. Set
\[
 B_{1,v}=g_v,\qquad
 B_{p,v}=\left(\sum_{u:(u,v)\in E}B_{p-1,u}\right)\wedge g_v.
\]
Every walk ending at $v$ has a unique previous endpoint $u$. Induction
therefore gives one term $g_{v_1}\wedge\cdots\wedge g_{v_p}$ for every
walk $(v_1,\ldots,v_p)$ ending at $v$. A repeated vertex makes that term
zero. After the final layer,
\[
 \sum_v B_{k,v}
 =\left(\sum_{P\in\calP_k}\Delta_P\right)e_{[k]},
 \qquad \Delta_P=\det(g_{v_1},\ldots,g_{v_k}).
\]
We append vectors on the right, so determinants follow the path order.
We use this convention throughout the graph algorithms.

The recurrence never stores a set of visited host vertices. It stores
exterior coordinates, and the alternating product handles distinctness.
This is why the algebra can replace explicit enumeration of paths.

\subsection{Why this sum is not a count}
The surviving determinants have different signs and magnitudes. They
can cancel. For example, take the three-vertex path $a-b-c$ with both
directions of each edge. Its two ordered three-vertex paths are
$(a,b,c)$ and $(c,b,a)$. Their determinants are negatives of one another,
so the displayed sum is zero for every choice of vectors.

Formal variables can prevent this cancellation for detection. Multiply
the contribution of role $p$ sent to $v$ by an indeterminate $y_{p,v}$.
Distinct paths have distinct products $\prod_p y_{p,v_p}$. Whenever a
path exists, its determinant coefficient is a nonzero polynomial in the
vector coordinates. This explains the usual algebraic detection idea:
evaluate a polynomial that is nonzero precisely when some desired term
survives. Its numerical value is still not the number of those terms.

To count, we would like each path to contribute one unit in expectation.
Squaring the final signed sum suggests a route, but also creates a new
problem: products $\Delta_P\Delta_Q$ for different paths. We next arrange
for these cross terms to have expectation zero, and then control the
variance of the resulting estimator.

\section{Designing the counting estimator}\label{sec:estimator}
\subsection{Why the exterior vectors are random}
\label{subsec:why-random-vectors}
One could assign a separate formal exterior basis vector to every host
vertex. This makes the distinctness rule transparent, but the coordinates
through degree $k$ would number $\sum_{p=0}^k\binom np$. We instead work
in the fixed $2^k$-dimensional algebra above.

The assignment $v\mapsto g_v\in\R^k$ preserves the identity that kills
repetitions. What it loses is a uniform numerical contribution from
different surviving paths. Independent standard Gaussian vectors restore
that uniformity in expectation: for every fixed simple path $P$,
$\E_g\Delta_P^2=k!$. Thus the same normalization works for every path.

Here a standard Gaussian vector means that its coordinates are independent
$N(0,1)$ variables. The determinant identity can already be seen from its
permutation expansion. In the square, two different permutation terms
leave a centered Gaussian entry occurring only once, and their expectation
is zero. Each of the $k!$ diagonal terms has expectation one. A later
moment calculation will also bound $\E_g\Delta_P^4$.

Random vectors alone do not solve cancellation: paths using the same
vertices can have equal determinants up to sign. We need a second random
ingredient to distinguish the paths themselves.

\subsection{A small matrix signature for each path}
Fix an auxiliary dimension $D$. Independently of all $g_v$, sample a
matrix $M_{p,v}\in\R^{D\times D}$ for every position $p\in[k]$ and host
vertex $v$. Its entries are independent $N(0,1/D)$ variables. Different
role--image pairs receive independent matrices.

Let $b_0=e_1\in\R^D$ be the first auxiliary basis vector. It is a fixed
starting and ending vector, not a graph vertex and not a random choice.
For $P=(v_1,\ldots,v_k)$, define
\begin{equation}\label{eq:path-signature-human}
 \chi(P)=\sqrt D\,b_0^\top M_{1,v_1}\cdots M_{k,v_k}b_0.
\end{equation}
The product has one matrix for every path position, including the first.
The row $b_0^\top$ is propagated from left to right. At the end,
multiplication by $b_0$ reads its first coordinate, and $\sqrt D$
normalizes that scalar projection.

For distinct paths, some position has different images. The matrix for
that role--image pair appears in one signature and not the other.
Conditioning on every other matrix makes the product linear in a centered
matrix, so $\E_M[\chi(P)\chi(Q)]=0$ for $P\ne Q$.
Also, $\E_M\chi(P)^2=1$. To see the normalization, condition on all but
the last matrix. For the preceding row $r$, the scalar $rM_{k,v_k}b_0$
has variance $\|r\|^2/D$. Multiplication by $D$ removes this factor,
and earlier normalized Gaussian matrices preserve the expected squared
row norm, starting from $\|b_0\|^2=1$.

The word ``signature'' does not mean that its value is $+1$ or $-1$.
It is a real random variable. The useful properties are its second
moments across paths and its fourth-moment bound across arbitrary sums.

\subsection{What one trial should output}
Our target amplitude and nonnegative estimator are
\begin{equation}\label{eq:unweighted-estimator-human}
 Z=\sum_{P\in\calP_k}\Delta_P\chi(P),\qquad Y=Z^2/k!.
\end{equation}
This display defines the desired result of a trial; it is not an instruction
to enumerate paths. The next section computes this sum by dynamic programming.

Condition first on all exterior vectors. The cross terms vanish over the
matrices, giving
\[
 \E_M[Z^2\mid(g_v)_v]=\sum_{P\in\calP_k}\Delta_P^2.
\]
Taking expectation over the vectors and dividing by $k!$ gives
$\E Y=N_k$. The accounting is now explicit: nonsimple walks contribute
zero identically, distinct paths have zero cross terms in expectation,
and each diagonal path term contributes one expected unit.

\subsection{Two overlapping paths}
Consider the directed graph with arcs $ab,bc,ad,dc$. There are exactly
two three-vertex paths, $P=(a,b,c)$ and $Q=(a,d,c)$.
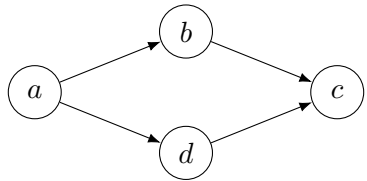
\begin{figure}[htbp]
\centering
\begin{tikzpicture}[v/.style={circle,draw,minimum size=7mm},>=Latex]
\node[v] (a) at (0,0) {$a$};
\node[v] (b) at (2,0.8) {$b$};
\node[v] (d) at (2,-0.8) {$d$};
\node[v] (c) at (4,0) {$c$};
\draw[->] (a)--(b); \draw[->] (b)--(c);
\draw[->] (a)--(d); \draw[->] (d)--(c);
\end{tikzpicture}
\caption{The two paths share their first and last vertices. Their random
signatures share the matrices for those roles as well.}
\label{fig:overlapping-paths}
\end{figure}
Their signatures are
$\sqrt D\,b_0^\top M_{1,a}M_{2,b}M_{3,c}b_0$ and
$\sqrt D\,b_0^\top M_{1,a}M_{2,d}M_{3,c}b_0$.
They are generally dependent: the first and last matrices are shared.
Nevertheless, conditioning on every matrix except $M_{2,b}$ proves that
their product has expectation zero. Likewise, their determinants share
$g_a$ and $g_c$ and need not be independent.

For this graph, $Z=\Delta_P\chi(P)+\Delta_Q\chi(Q)$. Conditional
expectation removes the cross term, leaving $\Delta_P^2+\Delta_Q^2$.
Each squared determinant has expectation $3!=6$, so $\E[Y]=2$.
The same calculation applies with more paths; it never asserts their
independence.

Why must the shared matrices actually be shared in the algorithm?
The state at a vertex represents a sum of partial paths. Applying one
matrix to this state distributes that matrix over the whole sum. This
allows the computation to stay small. Assigning fresh randomness to every
partial path would require distinguishing those partial paths explicitly.
The moment proof must accommodate the reuse that makes the dynamic
program efficient.

\subsection{Why cancellation in expectation is not enough}
An unbiased estimator can still fluctuate too much for efficient
approximation. We need a bound on $\E Y^2$, equivalently on the fourth
moment of $Z$. Pairwise cancellation alone does not give that bound.

For a small illustration, take independent fair signs $s_{p,0},s_{p,1}$
and form the scalar signature sum over all binary words:
$T=\prod_{p=1}^k(s_{p,0}+s_{p,1})$. Each factor is zero with probability
$1/2$, and otherwise is $+2$ or $-2$. Hence
$\E T^2=2^k$, $\E T^4=8^k$, and
$\E T^4/(\E T^2)^2=2^k$. This example concerns a coefficient sum;
it shows why orthogonality of word signatures alone does not establish
the uniform moment bound needed by our proof.

Matrices provide a dimension with which to control this accumulation.
One normalized Gaussian matrix multiplies the fourth moment of a fixed
vector norm by $1+2/D$, rather than by a fixed factor larger than one.
The full matrix-product lemma proves a corresponding bound for sums
with shared prefixes. Its loss is at most $3(1+2/D)^{k-1}$.
Taking $D=k$ makes this constant. The exterior determinants contribute
the remaining factor $K_k=(k+1)(k+2)/2$.

\begin{center}
\small
\begin{tabular}{@{}p{0.28\linewidth}p{0.64\linewidth}@{}}
\toprule
Ingredient & Its task\\
\midrule
Exterior product & Eliminate every repeated vertex exactly.\\
Gaussian vertex vectors & Make each surviving squared determinant have mean $k!$.\\
Role--image matrices & Cancel cross terms in expectation while controlling the fourth moment of the sum.\\
Independent trials & Turn the second-moment bound into a relative approximation.\\
\bottomrule
\end{tabular}
\end{center}

\section{Computing the path estimator}\label{sec:coordinate-implementation}
We now implement~\eqref{eq:unweighted-estimator-human}. All random
choices in this section are fixed for the duration of a trial.
Appendix~\ref{app:path-implementation} expands the implementation below,
including the individual table entries and all running-time sums.

\subsection{What a state stores}
The state $X_{p,v}$ has exterior degree $p$ and one auxiliary row of
length $D$ for each exterior coordinate:
\[
 X_{p,v}=\sum_{|S|=p}e_S\otimes X_{p,v}[S],
 \qquad X_{p,v}[S]\in\R^{1\times D}.
\]
Here $p$ is the number of path vertices already used and $v$ is the
current endpoint. Crucially, $S\subseteq[k]$ is a set of vector-coordinate
indices, not a set of host vertices. Many partial paths contribute to
the same row $X_{p,v}[S]$. No list of those paths is stored.

For instance, at degree two with $k=3$, the exterior coordinates are
$e_{\{1,2\}},e_{\{1,3\}},e_{\{2,3\}}$. A state at a fixed endpoint
stores three rows of length $D$. The contribution of a partial path
$(u,v)$ to the first row is the $\{1,2\}$ minor of $[g_u\ g_v]$
times $b_0^\top M_{1,u}M_{2,v}$. The other two rows use the other minors.

\subsection{The recurrence and the order of operations}
Initialize $X_{1,v}=g_v\otimes b_0^\top M_{1,v}$. At layer $p\ge2$, set
\begin{equation}\label{eq:unweighted-path-recurrence}
 U_{p,v}=\sum_{u:(u,v)\in E}X_{p-1,u},\qquad
 X_{p,v}=(U_{p,v}\wedge g_v)M_{p,v}.
\end{equation}
The wedge acts on the exterior coordinates; the matrix acts on their
auxiliary rows. They can be performed in either order because they act
on different factors. We use the displayed order consistently.
Because the matrix depends on $(p,v)$, all incoming contributions can
be added before applying it.

\begin{figure}[htbp]
\centering
\begin{tikzpicture}[box/.style={draw,rounded corners,align=center,
 minimum height=10mm,text width=37mm,font=\small},>=Latex]
\node[box,fill=softblue] (sum) {Sum over incoming arcs\\$U_{p,v}=\sum_{u:uv\in E}X_{p-1,u}$};
\node[box,fill=softgreen,right=9mm of sum] (wedge) {Append the vertex\\$U_{p,v}\wedge g_v$};
\node[box,fill=softorange,right=9mm of wedge] (matrix) {Update the signature\\$(U_{p,v}\wedge g_v)M_{p,v}$};
\draw[->,thick] (sum)--(wedge); \draw[->,thick] (wedge)--(matrix);
\end{tikzpicture}
\caption{One endpoint update. The same sampled matrix $M_{p,v}$ is used
for every exterior coordinate and every incoming contribution.}
\label{fig:layered-dp}
\end{figure}
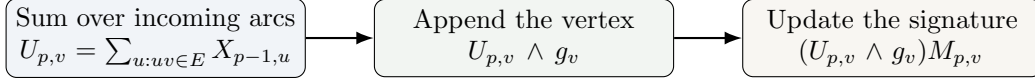

At degree $k$, only the coordinate $[k]$ remains. Form the row
$z=\sum_v X_{k,v}[[k]]$, where the double brackets indicate that unique
top coordinate. Output $Y=D(zb_0)^2/k!$. Thus the algorithm extracts a
scalar only after it has summed all endpoints.

\begin{lemma}[Expansion of the path recurrence]
\label{lem:human-path-expansion}
For every fixed choice of vectors and matrices,
\[
 X_{p,v}=\sum_{\substack{(v_1,\ldots,v_p)\text{ a walk}\\v_p=v}}
 (g_{v_1}\wedge\cdots\wedge g_{v_p})\otimes
 b_0^\top M_{1,v_1}\cdots M_{p,v_p}.
\]
Consequently the scalar $Z=\sqrt D\,zb_0$ is exactly the amplitude
in~\eqref{eq:unweighted-estimator-human}.
\end{lemma}
\begin{proof}
The statement for $p=1$ is the initialization. Every walk ending at $v$
with $p$ vertices is obtained uniquely by extending a walk ending at an
in-neighbour $u$. The first operation in~\eqref{eq:unweighted-path-recurrence}
adds these possibilities. The next two append $g_v$ and $M_{p,v}$ to
their respective products. Distributivity gives exactly the asserted sum.
Repeated host vertices give repeated exterior vectors, so their terms
vanish. At $p=k$, the surviving exterior coefficient is $\Delta_P$.
Summing endpoints and taking the normalized projection proves the claim.
\end{proof}

\subsection{The coordinate update, including its sign}
If $x$ has degree $p-1$, right multiplication by $g$ gives
\begin{equation}\label{eq:coordinate-wedge}
 (x\wedge g)[T]=\sum_{j\in T}(-1)^{p-\operatorname{pos}_T(j)}
                 g(j)x[T\setminus\{j\}],\qquad |T|=p.
\end{equation}
Here $\operatorname{pos}_T(j)$ is the position of $j$ in the increasing
order of $T$. The new $e_j$ starts at the right end, so it crosses the
$p-\operatorname{pos}_T(j)$ coordinates larger than $j$. This is the
source of the sign. For example, at degree two with $T=\{1,3\}$ the
row is $g(3)x[\{1\}]-g(1)x[\{3\}]$.

This map has $p\binom{k}{p}$ incidences, one for each pair $(T,j\in T)$.
It does not multiply every input coordinate by every output coordinate.
Applying it to rows of length $D$ costs
$\OO(Dp\binom{k}{p})$ operations.

\paragraph{One trial in coordinates.}
The following is the full dynamic program, with $D=k$.
\begin{enumerate}[leftmargin=1.7em,itemsep=3pt]
\item Sample the $g_v$ and all $M_{p,v}$ independently as specified above.
      Fix $b_0=e_1$.
\item For each vertex $v$ and $j\in[k]$, initialize
      $X_{1,v}[\{j\}]=g_v(j)b_0^\top M_{1,v}$.
\item For each $p=2,\ldots,k$, first set $U_{p,v}[S]=0$ for every
      endpoint $v$ and $(p-1)$-set $S$. Scan the arcs $(u,v)$ and add
      $X_{p-1,u}[S]$ to $U_{p,v}[S]$ for every such $S$.
\item At this layer, for each $v$ and $p$-set $T$, compute
      \[
       H_{p,v}[T]=\sum_{j\in T}(-1)^{p-\operatorname{pos}_T(j)}
                       g_v(j)U_{p,v}[T\setminus\{j\}],
       \qquad X_{p,v}[T]=H_{p,v}[T]M_{p,v}.
      \]
      Discard the preceding layer when its contributions have all been processed.
\item After layer $k$, form $z=\sum_vX_{k,v}[[k]]$ and return
      $D(zb_0)^2/k!$.
\end{enumerate}
Steps 3 and 4 are repeated together for each layer. Each sampled vector
or matrix is reused wherever its index recurs. In particular, we do not
resample a matrix for a different exterior coordinate. The polynomial
expansion in the lemma relies on that reuse.

The output may be zero on a graph with paths in a particular trial.
It is its expectation and controlled fluctuation across independent
trials that make it an approximate counter. On a graph with no such
path, every trial is exactly zero.
\section{Why the estimator concentrates}\label{sec:moments}
Unbiasedness tells us what a trial estimates. To use only polynomially many
trials, we also need a polynomial bound on its relative second moment.
Since the trial squares an amplitude, this is a fourth-moment question
about that amplitude. We first condition on the exterior vectors and
bound the matrix contribution. We then average over the exterior vectors.
This order lets us handle paths that share vertices and matrix choices.

\subsection{Gaussian variables, matrices, and the moment method}

We use Gaussian randomness in two places.  The vertex vectors $g_v\in\R^k$ create random exterior volumes.  The auxiliary matrices $M_{p,v}\in\R^{D\times D}$ create path signatures with small fourth moment.  This subsection records the definitions and moment identities used later.

\begin{definition}[Gaussian vectors and Gaussian matrices]
A random vector $h\in\R^d$ is a standard Gaussian vector, written $h\sim N(0,I_d)$, if its coordinates are independent $N(0,1)$ random variables.  More generally, $z\sim N(0,K)$ means that $z$ is a centered Gaussian vector with covariance matrix $K\succeq0$, i.e.
\[
  \E[z]=0,
  \qquad
  \E[zz^\top]=K.
\]
An $a\times b$ matrix $A$ is a standard Gaussian matrix if all its entries are independent $N(0,1)$ random variables.  A $D\times D$ matrix $G$ is a normalized Gaussian matrix if its entries are independent $N(0,1/D)$ random variables.
\end{definition}

The normalization $1/D$ is chosen so that multiplication by $G$ preserves Euclidean norm in expectation.  Indeed, for every fixed $y\in\R^D$,
\[
  Gy\sim N\left(0,\frac{\norm{y}^2}{D}I_D\right),
  \qquad
  \E\norm{Gy}^2=\norm{y}^2.
\]
Thus a product of independent normalized Gaussian matrices behaves like a random sign: it preserves the second moment.  The main extra point is that it also keeps the fourth moment under control.

\par\medskip\noindent\textbf{The Gaussian conditioning method.}\ 
The proofs repeatedly use the following template.  Condition on all randomness from earlier layers.  The next layer consists of fresh independent Gaussian matrices.  Therefore the new vector, or the rows of the new matrix, are centered Gaussian with an explicitly computable covariance.  Once the covariance is known, second and fourth moments follow from elementary Gaussian formulas.
\par\medskip

\begin{lemma}[Gaussian quadratic moments]\label{lem:gaussian-quadratic-prelim}
Let $z\sim N(0,K)$ in $\R^d$, where $K\succeq0$.  Then
\[
  \E\norm{z}^2=\tr K,
  \qquad
  \E\norm{z}^4=(\tr K)^2+2\tr(K^2).
\]
In particular, since $K\succeq0$,
\[
  \E\norm{z}^4\le 3(\tr K)^2.
\]
\end{lemma}

\begin{proof}
Write $z=K^{1/2}h$ with $h\sim N(0,I_d)$.  Then
\[
  \norm{z}^2=h^\top K h=\sum_i \lambda_i h_i^2,
\]
where $\lambda_i$ are the eigenvalues of $K$ and, by rotational invariance, the coordinates $h_i$ are taken in an orthonormal eigenbasis.  Hence
\[
  \E\norm{z}^2=\sum_i\lambda_i=\tr K.
\]
For the fourth moment,
\[
  \E\left(\sum_i\lambda_i h_i^2\right)^2
  =
  \sum_i\lambda_i^2\E h_i^4
  +2\sum_{i<j}\lambda_i\lambda_j\E h_i^2\E h_j^2.
\]
Using $\E h_i^2=1$ and $\E h_i^4=3$, this equals
\[
  3\sum_i\lambda_i^2+2\sum_{i<j}\lambda_i\lambda_j
  =
  \left(\sum_i\lambda_i\right)^2+2\sum_i\lambda_i^2
  =
  (\tr K)^2+2\tr(K^2).
\]
Finally, $\tr(K^2)\le(\tr K)^2$ for $K\succeq0$.
\end{proof}

\subsection{Matrix-product signatures}\label{sec:matrix-product-signs}

We now prove the fourth-moment bound for matrix-product signatures. The Gaussian
bound is credited to Rakhshan and Rabusseau~\cite{RR20}; we retain the complete
proof because the joint treatment of correlated prefixes is central to the
application to paths.

\begin{lemma}[Gaussian matrix-product signatures; Rakhshan--Rabusseau~\cite{RR20}]\label{lem:matrix-signs}
Let $r\ge 1$.  For each layer $i\in[r]$ and symbol $a\in\Sigma_i$, let $G_{i,a}\in\R^{D\times D}$ be independent Gaussian matrices with entries $N(0,1/D)$.  Fix unit vectors $\alpha,\beta\in\R^D$.  For a word $w=(w_1,\ldots,w_r)$, define
\[
  \chi(w)=\sqrt D\,\alpha^\top G_{r,w_r}G_{r-1,w_{r-1}}\cdots G_{1,w_1}\beta.
\]
Then, for every real coefficient vector $(a_w)$,
\[
  \E\left[\left(\sum_w a_w\chi(w)\right)^2\right]
  =
  \sum_w a_w^2,
\]
and
\[
  \E\left[\left(\sum_w a_w\chi(w)\right)^4\right]
  \le
  3(1+2/D)^{r-1}\left(\sum_w a_w^2\right)^2.
\]
\end{lemma}

The lemma is written for matrices acting on a column from right to left.
The path algorithm propagates a row from left to right; reversing the
layer indices gives its convention. The coefficients below are arbitrary
real numbers. In the application, we first fix the exterior vectors and
use the resulting determinants as these coefficients.

\begin{proof}
A single sum does not give a closed induction: separating its last symbol
produces a family of sums that share earlier matrices. We therefore prove
a bound for the combined squared norms of an arbitrary family. This lets
us carry all those sums through the induction without assuming that they
are independent.

We first prove a vector-valued strengthening.  Let $t\ge 0$ and let $H$ be a finite index set.  For every $h\in H$ and every word $w\in\Sigma_1\times\cdots\times\Sigma_t$, let $a_{h,w}\in\R$.  Define
\[
  X_h
  =
  \sum_{w\in\Sigma_1\times\cdots\times\Sigma_t}
  a_{h,w}G_{t,w_t}\cdots G_{1,w_1}\beta
  \in\R^D,
\]
with the convention that for $t=0$, $X_h=a_h\beta$.
We claim that
\begin{equation}\label{eq:vector-fourth}
  \E\left[\left(\sum_{h\in H}\|X_h\|^2\right)^2\right]
  \le
  (1+2/D)^t
  \left(\sum_{h,w}a_{h,w}^2\right)^2.
\end{equation}
For $t=0$, this holds with equality, because $\|\beta\|=1$.

Assume it holds for $t-1$.  Write a word as $(w',a)$, where $a\in\Sigma_t$ is the last symbol.  Define
\[
  Y_{h,a}
  =
  \sum_{w'} a_{h,w',a}G_{t-1,w'_{t-1}}\cdots G_{1,w'_1}\beta.
\]
Then
\[
  X_h=\sum_{a\in\Sigma_t}G_{t,a}Y_{h,a}.
\]
Condition on all vectors $Y_{h,a}$.  Let $X$ be the $D\times |H|$ matrix whose $h$-th column is $X_h$, and let $Y_a$ be the $D\times |H|$ matrix whose $h$-th column is $Y_{h,a}$.  Then
\[
  X=\sum_a G_{t,a}Y_a.
\]
The rows of $X$, conditioned on the $Y_a$'s, are independent centered Gaussian vectors in $\R^{|H|}$ with covariance
\[
  K=\frac1D\sum_a Y_a^\top Y_a.
\]
Let
\[
  S=\|X\|_F^2=\sum_h\|X_h\|^2.
\]
If $z\sim N(0,K)$, then
\[
  \E\|z\|^2=\tr K,
  \qquad
  \E\|z\|^4=(\tr K)^2+2\tr(K^2).
\]
Since $X$ has $D$ independent rows,
\[
  \E[S^2\mid Y]
  =D^2(\tr K)^2+2D\tr(K^2).
\]
Now
\[
  \tr K=\frac1D\sum_a\|Y_a\|_F^2,
\]
and, since $K\succeq 0$,
\[
  \tr(K^2)\le (\tr K)^2.
\]
Therefore
\[
  \E[S^2\mid Y]
  \le
  \left(1+\frac2D\right)
  \left(\sum_a\|Y_a\|_F^2\right)^2.
\]
Taking expectation and applying the induction hypothesis to the family indexed by pairs $(h,a)$ proves \cref{eq:vector-fourth}.

\paragraph{Why conditioning does not assume independent prefixes.}
The matrices at layer $t$ are independent of the sigma-field generated by
all matrices at layers $1,\ldots,t-1$. Every $Y_{h,a}$ is measurable with
respect to this earlier sigma-field. These vectors can be highly correlated:
different words may reuse exactly the same earlier matrices. The conditional
Gaussian calculation requires no independence between the $Y_{h,a}$'s.
After conditioning, they are a deterministic family. To apply the induction
hypothesis, take the new column index set to be
$H'=H\times\Sigma_t$ and the coefficients to be $a_{(h,a),w'}=a_{h,w',a}$.
Their squared coefficient mass is exactly $\sum_{h,w}a_{h,w}^2$.
The proof's large index set is an analytical device; the algorithm never
constructs a table indexed by all words.

We now derive the scalar fourth moment.  Take $t=r-1$.  For each final symbol $a\in\Sigma_r$, define
\[
  Y_a
  =
  \sum_{w'\in\Sigma_1\times\cdots\times\Sigma_{r-1}}
  a_{w',a}G_{r-1,w'_{r-1}}\cdots G_{1,w'_1}\beta.
\]
Then
\[
  \sum_w a_w\chi(w)
  =
  \sqrt D\,\alpha^\top\sum_a G_{r,a}Y_a.
\]
Conditioned on the $Y_a$'s, the vector $V=\sum_aG_{r,a}Y_a$ is centered Gaussian with covariance
\[
  \frac1D\left(\sum_a\|Y_a\|^2\right)I_D.
\]
Thus
\[
  \sqrt D\,\alpha^\top V
  \sim
  N\left(0,\sum_a\|Y_a\|^2\right),
\]
and hence
\[
  \E\left[\left(\sum_w a_w\chi(w)\right)^4\mid Y\right]
  =
  3\left(\sum_a\|Y_a\|^2\right)^2.
\]
Taking expectation and applying \cref{eq:vector-fourth} with $H=\Sigma_r$ gives the fourth-moment bound.

The second-moment identity follows similarly.  Conditioned on $Y_a$,
\[
  \E\left[\left(\sqrt D\,\alpha^\top V\right)^2\mid Y\right]
  =
  \sum_a\|Y_a\|^2.
\]
Iterating this identity backwards over the layers yields
\[
  \E\left[\left(\sum_w a_w\chi(w)\right)^2\right]
  =
  \sum_w a_w^2.
\]
This proves the lemma.
\end{proof}

\begin{remark}[Identification with the tensor-train projection]\label{rem:rr-credit}
Write $G_{i,a}=D^{-1/2}H_{i,a}$ with independent standard Gaussian entries.
For $r\ge2$, the boundary slices $H_{1,a}\beta$ and
$\alpha^\top H_{r,a}$ are standard Gaussian vectors. The intervening slices
are Gaussian matrices, and the overall normalization is $D^{-(r-1)/2}$.
Thus the contraction in \cref{lem:matrix-signs} is the single-output
tensor-train projection in~\cite{RR20}, with tensor order $r$ and rank $D$.
Its fourth-moment bound is the one stated above. The case $r=1$ is a single
Gaussian linear form. We do not claim this Gaussian moment lemma as new.
\end{remark}

For paths, we use words of length $k$ with role--image
pairs as symbols. The lemma allows arbitrary real coefficients and
therefore applies after conditioning on the exterior vectors, including
when different embeddings have shared vertices or prefixes.

\subsection{Exterior volumes}
\begin{lemma}[Determinant moments]\label{lem:det-moments}
Let $g_v\in\R^k$ be independent standard Gaussian vectors. For a $k$-set
$U$ let $\Delta_U$ be their determinant in any fixed column order, and put
$K_k=(k+1)(k+2)/2$. Then
\begin{equation}\label{eq:det-moments}
 \E\Delta_U^2=k!,\qquad \E\Delta_U^4=K_k(k!)^2.
\end{equation}
For arbitrary nonnegative masses $m_U$,
\begin{equation}\label{eq:weighted-volume-new}
 \E\left(\sum_U m_U\Delta_U^2\right)^2
 \le K_k(k!)^2\left(\sum_U m_U\right)^2.
\end{equation}
\end{lemma}
\begin{proof}
Apply Gram--Schmidt to the $k$ independent columns. Conditioned on previous
columns, the squared length of the next perpendicular component has the
$\chi^2_d$ distribution for $d=k,k-1,\ldots,1$. Its conditional law depends
only on $d$, so these squared lengths are independent. For a sum $Q$ of
$d$ independent squared standard Gaussians, $\E Q=d$ and
$\E Q^2=d(d+2)$. Their product is the squared determinant. Thus its first
two moments are $\prod_{d=1}^k d=k!$ and
$\prod_{d=1}^k d(d+2)=k!(k+2)!/2$, proving~\eqref{eq:det-moments}.
For any $U,U'$, Cauchy--Schwarz gives
$\E[\Delta_U^2\Delta_{U'}^2]\le K_k(k!)^2$.
Expand the square and use $m_Um_{U'}\ge0$.
\end{proof}

\subsection{Combining the two moments}
We can now justify the repetition bound promised during the construction.
Both kinds of dependence in the overlapping-path example are covered:
the matrix lemma treats shared matrix choices, and the determinant
bound treats shared vertex vectors.

\begin{theorem}[The unweighted path estimator]
\label{thm:human-path-estimator}
Set $D=k$, $K_k=(k+1)(k+2)/2$, and $C_k=3e^2K_k$. One trial of the
path algorithm satisfies
\[
 \E Y=N_k,\qquad \E Y^2\le C_kN_k^2.
\]
The median of $\OO(\log(2/\delta))$ independent group averages, each
using $\lceil8C_k\varepsilon^{-2}\rceil$ trials, is a relative
$\varepsilon$-approximation with probability at least $1-\delta$.
If $N_k=0$, every trial is zero.
\end{theorem}
\begin{proof}
Condition on all exterior vectors and apply Lemma~\ref{lem:matrix-signs}
to the coefficient family $a_P=\Delta_P$, assigning zero to other
role--image words. Its reversed layer convention is converted to the
forward products of~\eqref{eq:path-signature-human} by reversing the
layer indices. The lemma gives
\[
 \E_M Z^2=\sum_P\Delta_P^2,\qquad
 \E_M Z^4\le3(1+2/D)^{k-1}\left(\sum_P\Delta_P^2\right)^2.
\]
For each image set $U$, let $m_U$ count the ordered paths using exactly
that set. Their determinants differ only by column permutations, so
$\sum_P\Delta_P^2=\sum_Um_U\Delta_U^2$ and $\sum_Um_U=N_k$.
Lemma~\ref{lem:det-moments} now gives both claims after dividing by
$k!$ and $(k!)^2$, respectively, and using $(1+2/k)^{k-1}\le e^2$.

For $N_k>0$, a group average of $R$ independent trials has variance at
most $C_kN_k^2/R$. With $R=\lceil8C_k\varepsilon^{-2}\rceil$,
Chebyshev's inequality bounds its failure probability by $1/8$.
For an odd number $q$ of independent groups, the median can fail only
if at least $q/2$ groups fail. A union bound over these subsets gives
probability at most $2^q(1/8)^{q/2}=2^{-q/2}$. Choosing
$q=\OO(\log(2/\delta))$ proves the confidence guarantee.
If $N_k=0$, the expansion has no surviving term.
\end{proof}

\subsection{Running time and storage}
\begin{proposition}[Cost of the path dynamic program]
\label{prop:implementation-cost}
For auxiliary dimension $D\ge1$, one trial costs
\[
 \OO\bigl(2^k(mD+nkD+nD^2)+knD^2\bigr)
\]
arithmetic operations. Its working storage is
$\OO(nD\max_p\binom{k}{p}+nk+knD^2)$ words when all random matrices
are retained. With $D=k$, the resulting approximation algorithm takes
$2^k k^{\OO(1)}(n+m)\varepsilon^{-2}\log(2/\delta)$ operations.
\end{proposition}
\begin{proof}
At layer $p$, the arc scan adds rows in
$\OO(mD\binom{k}{p-1})$ operations. All endpoint wedge updates cost
$\OO(nDp\binom{k}{p})$, and their row--matrix multiplications cost
$\OO(nD^2\binom{k}{p})$.
Use $\sum_p\binom{k}{p}=2^k$ and
$\sum_p p\binom{k}{p}=k2^{k-1}$. Initialization and clearing tables
fit the bound, and sampling $kn$ matrices costs $\OO(knD^2)$.
Only two consecutive exterior layers and the current incoming sums are
needed, giving the table-storage bound. The remaining terms store the
vertex vectors and matrices. For the path algorithm alone, matrices can
instead be generated when their role--image pair is processed and then
discarded after use on all its coordinates. The stated bound also permits
retaining them, which will be useful when we reuse randomness across
separator assignments for trees.
Finally apply Theorem~\ref{thm:human-path-estimator}.
\end{proof}

Appendix~\ref{app:path-costs} spells out these costs step by step, including
subset indexing and table initialization.

The source of the base $2$ is now visible. Every transition adds one
exterior vector, with only $k2^{k-1}$ coordinate incidences over all
degrees. A dynamic program that explicitly carried pairs of exterior
coordinates would instead have
$\sum_p\binom{k}{p}^2=\binom{2k}{k}$ such coordinates. Our algorithm
does not form that tensor-square table. Its small auxiliary rows and
repeated random trials estimate the required squared mass.

The Gaussian model has served to describe the algorithm and prove its
moments. Section~\ref{sec:finite-new} gives one exact finite-randomness
route, after we introduce nonnegative weights.

%

\section{Finite randomness and bit complexity}\label{sec:finite-new}
The preceding analysis uses Gaussian random variables. We now give an
implementation using finitely many fair random bits and exact integer
and rational arithmetic. We first replace the Gaussian variables by a
finite distribution with the same moments through order four. We then
bound the bit lengths of the intermediate values and the number of
random bits used. The construction applies to the unweighted path,
forest, and graph-pattern estimators considered in this paper.

\subsection{Replacing Gaussian variables by a finite distribution}
\begin{definition}[Discrete coordinate law]\label{def:xi-new}
Let $\xi$ take values $-2,-1,0,1,2$ with respective probabilities
$1/12,1/6,1/2,1/6,1/12$.
\end{definition}
The distribution is symmetric about zero, so its odd moments vanish.
Its second and fourth moments are
\[
 \E[\xi^2]=\frac{2}{6}+\frac{8}{12}=1,
 \qquad
 \E[\xi^4]=\frac{2}{6}+\frac{32}{12}=3.
\]
Thus its moments of orders zero through four are $(1,0,1,0,3)$,
exactly those of a standard Gaussian variable.

\begin{lemma}[Moment matching]\label{lem:matching-new}
Replacing independent standard Gaussian variables by independent copies
of $\xi$ preserves the expectation of every polynomial whose degree in
each individual variable is at most four.
\end{lemma}
\begin{proof}
Expand the polynomial into monomials. Independence expresses the
expectation of each monomial as a product of scalar moments. Each
exponent is at most four, so every scalar moment agrees under the two
distributions. Summing proves the claim.
\end{proof}
The condition concerns the degree in each individual random variable,
not the total degree of the polynomial. Its total degree may be much
larger than four.

\subsection{An exactly computable estimator}
Recall that the Gaussian construction assigns a vector $g_v\in\R^k$
to each host vertex and a matrix $M_{i,v}\in\R^{D\times D}$ to each
pattern-role/image pair. The vector coordinates have distribution
$N(0,1)$, whereas the matrix entries have distribution $N(0,1/D)$.
The final scalar amplitude includes a factor $\sqrt D$, and one trial
returns $Y=Z^2/k!$.

To implement the estimator, we sample every coordinate of $g_v$ and
every entry of a matrix $H_{i,v}$ independently from the law of $\xi$.
We use $H_{i,v}$ in place of $M_{i,v}$ in the graph recurrence and omit
the final factor $\sqrt D$. All additions and multiplications in this
recurrence are then integer operations. We denote the resulting scalar
by $Z_0$.

\begin{theorem}[Discrete scalar estimator]\label{thm:discrete-new}
For the unweighted embedding computations of
Sections~\ref{sec:coordinate-implementation}, and \ref{sec:trees}, the estimator
\begin{equation}\label{eq:discrete-output-new}
 Y_0=\frac{Z_0^2}{k!D^{k-1}}
\end{equation}
has the same first two moments as the corresponding Gaussian estimator
$Y$. In particular, it is unbiased and satisfies the same second-moment
bound. If there are no embeddings, it is identically zero.
\end{theorem}
\begin{proof}
First let the entries of $H_{i,v}$ be independent standard Gaussians.
Then $M_{i,v}=D^{-1/2}H_{i,v}$ has the required normalized Gaussian
law. Every embedding term contains exactly $k$ matrix factors, one for
each pattern role. Including the final scalar projection gives
\[
 Z=\sqrt D\,(D^{-1/2})^k Z_0
   =D^{-(k-1)/2}Z_0.
\]
Consequently $Y=Z^2/k!=Z_0^2/(k!D^{k-1})$ under this Gaussian law.
This identity moves the normalization entirely into the final rational
denominator.

For arbitrary assignments of the vectors and matrices, the exterior
product eliminates every map with a repeated host vertex. Each
surviving term contains a determinant of $k$ distinct vertex vectors
and an ordered product of $k$ matrices with distinct role indices.
Expand both the determinant and the matrix product into scalar
monomials. Each vertex-coordinate variable occurs at most once in
such a monomial. Each matrix-entry variable also occurs at most once,
since its role occurs only once. Therefore $Z_0$ is individually
multilinear in the random scalar variables.

It follows that $Z_0^2$ has degree at most two in each random variable,
and $Z_0^4$ has degree at most four in each random variable.
Lemma~\ref{lem:matching-new} preserves the expectations of both
polynomials when all standard Gaussian entries are replaced by copies
of $\xi$. Dividing by the fixed normalization proves the two moment
identities. If no embedding exists, the exterior expansion is zero
for every assignment, so $Z_0=Y_0=0$.
\end{proof}
The argument concerns the polynomial after noninjective terms have
vanished. It does not assume that every intermediate table is
individually multilinear. The same sampled vectors and matrices are
reused throughout a trial, including across separator assignments.

\subsection{Bit lengths of intermediate values}
We take $D=k$ and assume $2\le k\le n$; the smaller and infeasible
cases are handled directly. Every sampled scalar is an integer of
absolute value at most two. We bound the integers created by the
direct graph recurrences.

Consider a stored entry containing the factors of $d\le k$ emitted
pattern roles. For a fixed partial map, an exterior coordinate is a
$d\times d$ minor. Expanding its determinant gives at most $d!$ terms,
each of absolute value at most $2^d$, so the minor has absolute value
at most $d!2^d$.
An entry of a product of $d$ raw matrices has at most
$D^{\max\{d-1,0\}}$ summands, each of absolute value at most $2^d$.
For $d=0$, the exterior and matrix factors are identities.

The expansion of a stored entry contains at most $n^k$ partial maps.
The recurrences count each such map once, with any boundary images
fixed by the state. Thus a uniform bound on the absolute value of a
stored integer is
\[
 n^k k!4^kD^{k-1}.
\]
With $D=k$, the required bit length is
\begin{equation}\label{eq:bit-bound-new}
 B_0=\OO\bigl(k\log(n+2)+k\log(k+2)\bigr).
\end{equation}
This argument also bounds sums over separator assignments: an
embedding has a unique restriction to the separator, so those sums
do not introduce additional copies of an embedding.

We must also bound temporary values before cancellations occur.
Multiplying two stored entries uses at most $2B_0+\OO(1)$ bits.
At an exterior product, a scalar output coordinate sums at most
$3^kD$ scalar products, adding only $\OO(k+\log(D+1))$ bits.
The vertex and adjacency sums add at most $\OO(\log(n+2))$ bits
per such aggregation. Equivalently, their temporary values can be
bounded by the sum of the absolute values of their expanded terms.
Hence the temporary products and partial sums also have polynomial
bit length.

The denominator $C=k!D^{k-1}$ in~\eqref{eq:discrete-output-new} has
$\OO(k\log(k+2))$ bits. All trials use this same denominator.
An average of $r$ trials can therefore be computed as
\[
 \frac{\sum_{j=1}^{r}Z_{0,j}^2}{rC},
\]
whose numerator and denominator require only
$\OO(B_0+\log(r+1))$ bits. Taking the median of group averages uses
exact rational comparisons. Thus the estimator and its amplification
admit exact arithmetic with polynomial bit overhead.

\subsection{Sampling with a bounded number of random bits}
The law of $\xi$ can be sampled exactly by rejection. Draw four fair
bits, interpret them as an integer in $\{0,\ldots,15\}$, and repeat
if the integer is at least $12$. Among the twelve accepted values,
assign six to zero, two to each of $1,-1$, and one to each of $2,-2$.
Each attempt succeeds with probability $3/4$ and gives the required
law upon acceptance.

This sampler has constant expected cost, but it has no deterministic
bound on the number of attempts. We obtain a bounded-time algorithm
by allowing a small additional failure probability. Let $R$ be the
total number of trials in the amplified estimator, chosen using
statistical failure probability $\delta/2$. Sampling all vertex
vectors and all role--image matrices requires at most
\[
 T=Rkn(1+D^2)
\]
scalar draws. The same samples are used for all states and separator
assignments within a trial.

For every scalar draw, allow at most
\[
 J=\left\lceil\log_2\frac{2(T+1)}{\delta}\right\rceil
\]
attempts. If any draw exhausts this allowance, stop and return zero.
A fixed draw exhausts its allowance with probability $(1/4)^J$, so
by a union bound the probability that any draw does so is at most
\[
 T(1/4)^J\le T2^{-J}\le\delta/2.
\]
Couple this algorithm with the uncapped exact sampler using the same
random bits. Unless an allowance is exhausted, their outputs agree.
The uncapped estimator fails with probability at most $\delta/2$;
therefore the capped algorithm fails with probability at most $\delta$.
This argument does not require independence conditional on completing
all sampling calls. If the true count is zero, both a completed trial
and a sampling abort return zero.

The implementation uses at most $4TJ$ fair random bits. Together with
\eqref{eq:bit-bound-new} and the trial bound
$R=\OO(k^2\varepsilon^{-2}\log(2/\delta))$, this proves polynomial
bit overhead in $k$, $\log(n+2)$, $\log(1/\varepsilon)$, and
$\log(2/\delta)$ relative to the stated arithmetic-operation bounds.
All recurrence values are computed exactly; no numerical rounding is
needed.
%

\section{Trees through small-component separators}\label{sec:trees}
A path extension inserts one vertex, so one exterior factor always has
degree one. A tree can branch into two large subtrees. Joining their
partial embeddings requires multiplying two exterior elements, and a
direct multiplication may examine $3^k$ disjoint pairs of coordinate
sets. Thus the path running-time argument does not extend merely by
rooting the tree and multiplying child tables.

We remove a small set of pattern vertices so that every remaining
component is small, and enumerate the images of the removed vertices.
For a fixed assignment, an edge from a component to the separator
restricts the possible image of its endpoint in the component.
We compute the component tables separately and multiply them in a
fixed order. Each merge has a factor of small exterior degree.

The estimator counts embeddings of a specified tree or forest.
All separator assignments use the same random vectors
and matrices. We add their contributions before squaring, and prove
that the resulting sum contains each embedding exactly once.
The same matrix and determinant moment bounds used for paths then
apply to this sum.

We first treat an undirected tree $T$ on $k$ vertices in a simple
undirected host $G$ with $n$ vertices and $m$ edges. The extensions to
forests and orientations are given below. If $k>n$, return zero;
the empty pattern has one embedding, and a one-vertex pattern has
$n$ embeddings. We therefore assume $2\le k\le n$ in the construction.

\subsection{Multiplying component states}
Paths could propagate a row because their factors arrived sequentially.
A component must be composable with other components on either side, so
we retain a full matrix in each exterior coordinate until the final
boundary contraction. This is the same forward matrix product as before,
with more information retained in a partial state.

We use matrix-valued exterior elements in
$\mathcal A_D=\Lambda(\R^k)\otimes\Mat_D(\R)$, with product
\begin{equation}\label{eq:algebra-product}
 (x\otimes A)\star(y\otimes B)=(x\wedge y)\otimes AB.
\end{equation}
This algebra is associative and generally noncommutative. Every recurrence
using this product specifies a fixed multiplication order. Every summand follows the same fixed order of pattern vertices.

\begin{lemma}[Exterior multiplication cost]\label{lem:small-product}
An exterior element can be multiplied by a vector in $\OO(k2^k)$ scalar
operations. If one factor has degree at most $h\le k$, a product in
$\mathcal A_D$ can be computed using
\begin{equation}\label{eq:small-product-cost}
 \OO\left(kD^3\,2^k\sum_{j=0}^h\binom{k}{j}\right)
\end{equation}
operations. When $h\le k/2$, this is at most
$k^{\OO(1)}D^3 2^{(1+\HH(h/k))k}$, where
$\HH(x)=-x\log_2x-(1-x)\log_2(1-x)$ and $\HH(0)=0$.
The same bound holds if the small-degree factor is on the left.
\end{lemma}
\begin{proof}
For multiplication by a vector, each nonzero incidence has the form
$(A,j)\mapsto A\cup\{j\}$ with $j\notin A$, giving at most $k2^{k-1}$
incidences. If $B$ indexes the factor of degree at most $h$, there are at
most $\sum_{j\le h}\binom{k}{j}2^{k-j}$ disjoint pairs $(A,B)$. Multiply
their matrix coefficients by the elementary $D^3$ algorithm and accumulate
with the sign in~\eqref{eq:wedge-product}. The sign can be computed in
$\OO(k)$ index operations by scanning the coordinates and maintaining
the number of earlier elements of the other set. This accounts for the
factor $k$ in~\eqref{eq:small-product-cost}.
For the entropy estimate, take $0<h\le k/2$, put $x=h/k$, and expand $(x+(1-x))^k$. For $j\le h$,
$x^j(1-x)^{k-j}\ge x^h(1-x)^{k-h}$; hence
$\sum_{j\le h}\binom{k}{j}\le[x^h(1-x)^{k-h}]^{-1}=2^{k\HH(x)}$.
The cases $h=0$ and $x=1/2$ follow directly.
\end{proof}

\subsection{The separator and the exponential-base parameter}
\begin{lemma}[Small components with two boundary vertices]\label{lem:separator-new}
Let $T$ be a tree on $k$ vertices and let $c\ge2$ be an integer. In polynomial
time one can find $S\subseteq V(T)$ of size at most $2c+2$ such that every
component $C$ of $T-S$ has at most $h=\lfloor k/c\rfloor$ vertices and
$|N_T(C)|\le2$.
\end{lemma}
\begin{proof}
Root the tree. Process it from the leaves upwards, keeping the unmarked
vertices still attached to each processed root. If this residual subtree
has more than $h$ vertices, mark its root and reset its residual size to
zero. Each mark can be charged to more than $h$ previously uncharged
vertices. Thus there are at most $\lfloor k/(h+1)\rfloor\le c$ marks,
and every component left after removing the marks has size at most $h$.
Add the root and close the marked set under least common ancestors.
If $S_0$ is the marked set together with the root, its least-common-ancestor
closure has size at most $2|S_0|\le2c+2$: the additional vertices are
branching vertices of the minimal rooted subtree connecting $S_0$.
If a remaining component had three boundary vertices, at least two would
lie below it in distinct branches, and their least common ancestor would
belong to that component. This contradicts closure. Removing further
vertices does not increase component sizes.
\end{proof}
This is the separator statement used in~\cite[Lemma 5.18]{FLPS16}; the
integer threshold above also covers $h=0$.

\paragraph{Which separator properties the recurrence needs.}
The two-boundary version is retained to make the connection with the
representative-family construction explicit. For the recurrence below,
only the component-size bound is needed: after fixing all images of $S$,
every edge to the separator becomes a restriction on the allowed image
of its other endpoint.
Thus one may also use the greedy marked set before adding the root and
least common ancestors. It has size at most $c$, and all remaining
components still have at most $h$ vertices. The stated theorem holds with
either choice. This observation will also justify the forest extension.

The parameter $c$ determines the tradeoff. A larger $c$ makes the remaining
components smaller and their exterior products cheaper, but increases the
number of separator images to enumerate. We choose it as a function of
$\eta$ alone, so this enumeration remains polynomial in the host size for
every fixed $\eta$.

For $0<\eta\le1$, fix an integer $c\ge2$ with
\begin{equation}\label{eq:c-choice}
 \HH(1/c)\le\log_2(1+\eta/2).
\end{equation}
Such a $c$ exists because $\HH(x)\to0$ as $x\to0$.
Moreover $\HH(1/c)\le(\log_2 c+\log_2 e)/c$, so one can take
$c=\OO(\eta^{-1}\log(2/\eta))$. Since $h/k\le1/c$,
Lemma~\ref{lem:small-product} bounds each small-factor multiplication by
$k^{\OO(1)}D^3(2+\eta)^k$ operations. For $\eta>1$ use $\eta=1$.

\subsection{Pinning the separator}
For a concrete example, remove the centre $r$ of a seven-vertex tree
whose three branches are the edges $a-b$, $c-d$, and $e-f$,
with $r$ adjacent to $a$, $c$, and $e$.
Once $r$ is sent to a host vertex $x$, the images of $a$, $c$, and $e$
must all be neighbours of $x$. The three component tables can then
be computed separately. Their chosen images may overlap, so the
exterior product is still needed when combining the tables: any
repeated host vertex makes the corresponding term zero.

Fix a separator $S$. We enumerate the assignments
\begin{equation}\label{eq:separator-amplitude}
 \mathcal S=
 \{\sigma:S\to V(G):\sigma\text{ is injective and preserves
 every edge of }T[S]\}.
\end{equation}
For $i\notin S$ and $\sigma\in\mathcal S$, define its allowed images by
\begin{equation}\label{eq:modified-unary}
 L_i^\sigma=
 \{v\in V(G)\setminus\sigma(S):
       v\sigma(j)\in E(G)\text{ for every }j\in S\cap N_T(i)\}.
\end{equation}
The assignment $\sigma$ checks edges entirely within the separator.
The sets $L_i^\sigma$ check edges from the separator to a component
and exclude separator images. Edges within a component will be
checked in its recurrence. These three classes partition the edges
of the pattern.

List $S$ in a fixed order, list the components as $C_1,\ldots,C_\ell$,
and fix a root and child order within each component.
The global order of pattern vertices is $S$ followed by the preorder
traversals of $C_1,\ldots,C_\ell$. Write it as $i_1,\ldots,i_k$.
It is independent of $\sigma$ and of all choices of host images.

Independently sample a standard Gaussian vector $g_v\in\R^k$ for
each host vertex, and a matrix $M_{i,v}\in\R^{D\times D}$ with
independent $N(0,1/D)$ entries for each pattern vertex $i$ and host
vertex $v$. The matrices and vertex vectors are mutually independent.
Set $A_{i,v}=g_v\otimes M_{i,v}$.

For a rooted component and a vertex $i$ with ordered children
$j_1,\ldots,j_d$, compute
\begin{equation}\label{eq:tree-recurrence-new}
 Q^\sigma_{i,v}=
 \begin{cases}
 A_{i,v}\star
 \displaystyle\mathop{\bigstar}_{a=1}^d
       \left(\sum_{u\in N_G(v)}Q^\sigma_{j_a,u}\right),
       &v\in L_i^\sigma,\\[2mm]
 0,&v\notin L_i^\sigma.
 \end{cases}
\end{equation}
The state sums the contributions of maps from the subtree rooted
at $i$ with $i$ sent to $v$. The neighbour sum checks the edge
$ij_a$. Products over children follow their fixed order.
The empty product is $1\otimes I_D$, so a leaf state is $A_{i,v}$
when $v$ is allowed and zero otherwise.

Let $r_C$ be the root of component $C$ and put
$B_C^\sigma=\sum_v Q^\sigma_{r_C,v}$.
Combine the component answers and sum over separator assignments:
\begin{equation}\label{eq:tree-global-new}
 \begin{aligned}
 A&=\sum_{\sigma\in\mathcal S}
       \left(\mathop{\bigstar}_{i\in S}A_{i,\sigma(i)}\right)
       \star B_{C_1}^\sigma\star\cdots\star B_{C_\ell}^\sigma,\\
 Z&=\sqrt D\,b_0^\top[A]_{[k]}b_0,
 \qquad Y=Z^2/k!,\qquad b_0=e_1\in\R^D.
 \end{aligned}
\end{equation}
Here $[A]_{[k]}$ is the matrix coefficient of the unique top-degree
exterior basis element. The same sampled $g_v$ and $M_{i,v}$ are used
for every $\sigma$. In particular, we sum the scalar contributions of
all assignments before taking the square.

\subsection{Why the estimator counts every embedding}
For an embedding $f$ of $T$ into $G$, define
\[
 \Delta_f=\det[g_{f(i_1)}\ \cdots\ g_{f(i_k)}],
 \qquad
 \chi(f)=\sqrt D\,b_0^\top
 M_{i_1,f(i_1)}\cdots M_{i_k,f(i_k)}b_0.
\]
Both products follow the fixed global order of pattern vertices.

\begin{lemma}[Exact tree expansion]\label{lem:tree-expansion-new}
The scalar in~\eqref{eq:tree-global-new} satisfies
\begin{equation}\label{eq:tree-unweighted-expansion}
 Z=\sum_{f\in\Emb(T,G)}\Delta_f\chi(f).
\end{equation}
Each embedding occurs exactly once.
\end{lemma}
\begin{proof}
For a fixed $\sigma$, induction on a component subtree expands
\eqref{eq:tree-recurrence-new} over its edge-preserving maps, with
the root image fixed and every vertex image in its allowed set.
At a leaf the claim is the initialization. At an internal vertex,
the neighbour sums list the possible images of each child, and the
ordered product combines one choice from each child subtree.
Each map has a unique such decomposition, so it appears once.
Its term contains the exterior and matrix factors of exactly those
subtree vertices, in preorder.

These maps may repeat a host vertex, either within a component or
between components. In the complete exterior product, every such
term contains the same vertex vector twice and is zero.
The sets $L_i^\sigma$ already exclude collisions with separator
images. Thus the surviving terms are precisely the embeddings.
Every embedding has a unique restriction $f|_S=\sigma$ and unique
component restrictions, so it occurs once in the sum over $\sigma$.
Its top exterior coefficient is $\Delta_f$, and its matrix product
has exactly the order defining $\chi(f)$. This proves the expansion.
\end{proof}

\begin{proposition}[Moments of the tree estimator]
\label{prop:tree-unweighted-moments}
Let $N_T=|\Emb(T,G)|$, set $D=k$, and put
$K_k=(k+1)(k+2)/2$. Then
\[
 \E Y=N_T,
 \qquad
 \E Y^2\le 3e^2K_kN_T^2.
\]
If $N_T=0$, every trial is zero.
\end{proposition}
\begin{proof}
Condition on all vertex vectors. An embedding is determined by its
word of images in the fixed pattern order. Different embeddings
therefore give different words, so Lemma~\ref{lem:matrix-signs}
applies to the coefficient family $(\Delta_f)_f$. Reversing layer
indices converts its matrix order to the forward convention above.
It gives
\[
 \E_M[Z^2\mid g]=\sum_f\Delta_f^2,
 \qquad
 \E_M[Z^4\mid g]\le
 3(1+2/D)^{k-1}\left(\sum_f\Delta_f^2\right)^2.
\]
Each embedding uses $k$ distinct vertex vectors, so
Lemma~\ref{lem:det-moments} gives $\E_g\Delta_f^2=k!$.
For any two embeddings $f,f'$, Cauchy--Schwarz and the same lemma give
\[
 \E_g[\Delta_f^2\Delta_{f'}^2]\le K_k(k!)^2,
\]
even when their images overlap. Expand the square in the fourth-moment
bound, take expectation over the vectors, and divide by $(k!)^2$.
Using $(1+2/k)^{k-1}\le e^2$ proves the second-moment bound.
The first-moment identity follows by dividing $\E Z^2=k!N_T$ by
$k!$. If $N_T=0$, the exact expansion is empty.
\end{proof}

The same averaging and median argument used for paths now gives a
relative $\varepsilon$-approximation with failure probability at
most $\delta$ using
$\OO(k^2\varepsilon^{-2}\log(2/\delta))$ independent trials.
This trial count applies to the complete sum over separator
assignments. There is no union bound over those assignments and
no additional statistical loss depending on their number.

\subsection{Running time}
\begin{proof}[Proof of Theorem~\ref{thm:intro-tree} for trees]
Inside a component, every intermediate exterior degree is at most
$h$. Evaluate products with several children sequentially: every
factor being merged has degree at most $h$. The same is true when
assembling component answers in~\eqref{eq:tree-global-new}.
Build the separator product one vertex at a time. These vector
insertions cost $2^k k^{\OO(1)}D^3$ operations, also when $h=0$.
By the choice of $c$ in~\eqref{eq:c-choice}, each small-factor
product costs at most $(2+\eta)^k k^{\OO(1)}D^3$ operations.

There are at most $n^{|S|}$ injective assignments to the separator.
For each, check the edges of $T[S]$ and discard the assignment if
an edge is not preserved. If $S\ne\varnothing$, build a host adjacency
matrix once in $\OO(n^2+m)$ time. This permits constant-time edge
tests when checking separator assignments and forming the allowed
sets $L_i^\sigma$. Since $|S|\ge1$, this preprocessing is absorbed
by $n^{|S|}(n+m)$. If $S=\varnothing$, there are no separator-edge
tests and adjacency lists suffice. For each assignment the allowed
sets can be computed in $k^{\OO(1)}n$ time.

There are $\OO(k)$ child merges per host vertex and $\OO(k)$
neighbor aggregations per separator assignment. The aggregations
scan host adjacency lists and add matrix-valued exterior coordinates;
their total cost per assignment is at most
$2^k k^{\OO(1)}D^2(n+m)$.
Consequently one complete trial costs
\begin{equation}\label{eq:tree-precise-time}
 (2+\eta)^k k^{\OO(1)}D^3 n^{|S|}(n+m).
\end{equation}
This bound includes all separator assignments. Component maps are
computed by the recurrence, not enumerated individually.
Sampling all vectors and matrices also fits the bound.

Set $D=k$ and apply Proposition~\ref{prop:tree-unweighted-moments}
and the amplification argument above. Since $|S|\le2c+2$ and
$c=\OO(\eta^{-1}\log(2/\eta))$ for $0<\eta\le1$, the resulting
running time is
\[
 (2+\eta)^k n^{\OO_\eta(1)}
 \varepsilon^{-2}\log(2/\delta),
\]
where the exponent of $n$ is
$\OO(\eta^{-1}\log(2/\eta))$.
Here we absorb polynomial factors in $k$ and the factor $n+m$
into the power of $n$, using $k\le n$ and simplicity of the host.
Section~\ref{sec:finite-new} replaces the Gaussian construction by
finite randomness and exact integer and rational arithmetic with
polynomial bit overhead.
\end{proof}

\subsection{Forests, orientations, and copies}
For a forest $F$, root each connected component and perform greedy
marking across all of them with the common threshold
$h=\lfloor k/c\rfloor$. Each mark is charged to at least $h+1$
previously uncharged vertices, so there are at most
$k/(h+1)\le c$ marks in total. Every component left after removing
the marked set has at most $h$ vertices.
The two-boundary property is unnecessary, so no least-common-ancestor
closure is needed.

Use this marked set as $S$ and apply the same allowed-set definition,
component recurrence, and global product, with $F$ in place of $T$.
Fix the component and traversal orders once for all assignments.
Exterior multiplication enforces distinctness both within and between
the original forest components, including isolated vertices.
The expansion and moment proofs are unchanged, with
$N_F=|\Emb(F,G)|$ in place of $N_T$.
This proves the forest case of Theorem~\ref{thm:intro-tree}.
Using only greedy marks for either trees or forests gives the explicit
host factor $n^c(n+m)$ in~\eqref{eq:tree-precise-time}.

For an oriented forest in a directed host, choose the separator and
root the components in the underlying undirected pattern.
Require $\sigma$ to preserve the directions of arcs within $S$.
For an arc $i\to j$ with $i\notin S$ and $j\in S$, the allowed
set requires $(v,\sigma(j))\in E(G)$; for an arc $j\to i$, it
requires $(\sigma(j),v)\in E(G)$.
At a parent $i$ and child $j$ in a component, replace the neighbour
sum by a sum over out-neighbours of $v$ if the pattern arc is
$i\to j$, and over in-neighbours if it is $j\to i$.
All products still follow the same fixed traversal order.
The expansion and running-time arguments therefore remain valid,
with $m$ counting host arcs.

Finally, embeddings distinguish the vertices of the specified pattern.
Every non-induced subgraph isomorphic to $F$ has exactly
$|\Aut(F)|$ embeddings of $F$ onto it. Dividing the estimated
embedding count by $|\Aut(F)|$ therefore estimates the number of
such copies with the same relative error and success probability.
For an oriented pattern, use its direction-preserving automorphisms.

For trees and forests, this automorphism number is computable in
polynomial time. For a rooted tree, group isomorphic child subtrees,
multiply their automorphism numbers, and multiply by the factorial
of each group's multiplicity. For an oriented tree, record the
parent--child edge direction in each child type. For an unrooted
tree, use its one or two underlying tree centres and include a
centre swap exactly when it preserves the pattern. For a forest,
also group isomorphic connected components and account for their
permutations. These integers have polynomial bit length, since
$|\Aut(F)|\le k!$.
This division counts subgraphs specified by their vertices and edges;
it does not in general count distinct supporting vertex sets.
%
\section{Conclusions and further questions}
We gave randomized approximation algorithms for counting paths and
embeddings of specified forests. For paths on $k$ vertices, the
running time has a $2^k$ dependence on $k$. For every fixed
$\eta>0$, the forest algorithm has running time
$(2+\eta)^k n^{\OO_\eta(1)}\varepsilon^{-2}\log(2/\delta)$.
The construction also gives algorithms for patterns with supplied
path and tree decompositions.

The algorithms combine exterior algebra with random matrix
signatures. Exterior multiplication eliminates maps that repeat a
host vertex. The signatures make cross terms between distinct
embeddings vanish in expectation, and the moment bounds ensure
that polynomially many trials suffice. For forests, a small-component
separator allows the embedding sum to be evaluated using exterior
products with a small-degree factor. The estimators admit
finite-random-bit implementations with polynomial bit overhead.

Two questions remain. First, can the number of embeddings of a given
forest on $k$ vertices be approximated in
$2^k\poly(n,\varepsilon^{-1})$ time with constant success probability?
Our bound approaches base $2$, but the polynomial exponent in $n$
increases as $\eta$ decreases. It therefore does not give this
running time. Second, can the exponential space used by our
algorithms be reduced to polynomial space while retaining their
running-time bounds and approximation guarantees?

\section*{Acknowledgements}
Saket would like to thank Prof. R. Balasubramanian for suggesting the use of
Gaussian random variables in the estimator.

\paragraph{Use of generative AI.}
The authors used OpenAI ChatGPT in the preparation of this manuscript to assist with
(i)~conducting a literature survey,
(ii)~tightening and improving the language and exposition throughout the paper,
(iii)~suggesting formulations for parts of the introduction,
(iv)~generating preliminary drafts of proofs based on proof scaffolding, proof strategies,
and technical directions provided by the authors,
(v)~preparing and drawing figures, and
(vi)~preparing an initial draft of this disclosure statement.
The tool materially assisted the preparation and exposition of the manuscript across
multiple sections. The mathematical ideas, proof strategies, and technical direction
were supplied and developed by the authors. The authors carefully checked and take
full responsibility for the correctness and originality of all content, including all
mathematical statements, proofs, and references.

\bibliographystyle{alpha}
\bibliography{counting}
\begingroup
\small
\setlength{\cftbeforesecskip}{0pt}
\setlength{\cftbeforepartskip}{6pt}
\renewcommand{\cftsecfont}{\small}
\renewcommand{\cftsecpagefont}{\small}
\renewcommand{\cftpartfont}{\small\bfseries}
\renewcommand{\cftpartpagefont}{\small\bfseries}
\endgroup

\appendix
\section{A detailed implementation of the path dynamic program}
\label{app:path-implementation}
This appendix expands the implementation in
Section~\ref{sec:coordinate-implementation} and the cost proof in
Proposition~\ref{prop:implementation-cost}. We follow one trial from its
random choices to its final scalar, including the individual table entries.
The goal is to make the source of the $2^k$ running time visible at the
level of the actual loops. The main text contains the estimator and its
proofs; this appendix supplies a slower account of how to compute it.

Throughout, a path has $k$ vertices and $k-1$ arcs. The host graph is
simple and directed, with $n$ vertices and $m$ arcs, and $2\le k\le n$.
We use precisely the convention of the main text: append exterior vectors
on the right, store auxiliary \emph{rows}, and use a matrix $M_{p,v}$
indexed by the path position and its image. There is a matrix at position
one as well. We first explain unit weights and return to local weights
in Appendix~\ref{app:path-local-weights}.

\subsection{What each table index means}
Fix a layer $p$ and an endpoint $v$. The state is
$X_{p,v}=\sum_{|S|=p}e_S\otimes X_{p,v}[S]$, where
$X_{p,v}[S]\in\R^{1\times D}$. One scalar entry can therefore be
addressed by $(p,v,S,t)$, with $t\in[D]$. These indices have different
meanings:
\begin{center}
\begin{tabular}{@{}ll@{}}
\toprule
Index & Meaning\\
\midrule
$p$ & Number of vertices in the partial walk.\\
$v\in V(G)$ & Its last host vertex.\\
$S\subseteq[k]$, $|S|=p$ & An exterior-coordinate index set.\\
$t\in[D]$ & A coordinate of the auxiliary row.\\
\bottomrule
\end{tabular}
\end{center}
In particular, $S$ is not a set of visited graph vertices. Its elements
refer to the $k$ coordinates of the assigned vectors $g_v$. This distinction
explains why the table has $\binom{k}{p}$ rows at an endpoint, rather than
$\binom{n}{p}$ rows.

For $k=3$ and $D=2$, the complete shape at one endpoint is
\begin{center}
\begin{tabular}{@{}clc@{}}
\toprule
Layer & Exterior row labels & Number of scalar entries\\
\midrule
$1$ & $\{1\},\{2\},\{3\}$ & $3\cdot2=6$\\
$2$ & $\{1,2\},\{1,3\},\{2,3\}$ & $3\cdot2=6$\\
$3$ & $\{1,2,3\}$ & $1\cdot2=2$\\
\bottomrule
\end{tabular}
\end{center}
The number of walks contributing to these entries can be much larger
than the number of entries. We add their numerical contributions into the
same rows; we never attach a separate record to each walk.

\paragraph{What contributes to a row?}
Let $P=(v_1,\ldots,v_p)$ be a walk ending at $v$, and define the auxiliary
row $r(P)=b_0^\top M_{1,v_1}\cdots M_{p,v_p}$, where $b_0=e_1\in\R^D$.
For $S=\{s_1<\cdots<s_p\}$, let $G_{P,S}$ have entries
$(G_{P,S})_{a,b}=g_{v_b}(s_a)$. The contribution of $P$ to $X_{p,v}[S]$
is exactly $\det(G_{P,S})r(P)$, by
Lemma~\ref{lem:wedge-coordinate-prelim}. Thus the exterior coordinate
supplies a minor, while the auxiliary row supplies an ordered matrix
product. This is an interpretation of the entry, not an enumeration
procedure for computing it.

\subsection{What is sampled, and what is reused?}
For each vertex $v$, sample $g_v\in\R^k$ with independent $N(0,1)$
coordinates. For each $(p,v)\in[k]\times V(G)$, sample a $D\times D$
matrix $M_{p,v}$ with independent $N(0,1/D)$ entries. These choices are
mutually independent. They are fixed for the entire trial.

All paths visiting $v$ use the same exterior vector $g_v$, regardless
of position. This reuse makes a repeated visit contribute a repeated
vector, which forces its exterior product to vanish. The matrices have
a different indexing rule: position $p$ at vertex $v$ always uses
$M_{p,v}$, and different role--image pairs have independent matrices.
For a fixed $(p,v)$, that one matrix is used on every exterior row and
every incoming contribution. Resampling it for individual rows would
change the polynomial computed by the recurrence.

The boundary vector $b_0=(1,0,\ldots,0)^\top$ is deterministic. It is
not a host vertex or a stored subset. The initial auxiliary row is
$b_0^\top M_{1,v}$, the first row of $M_{1,v}$. A partial walk carries
this row forward by multiplying the subsequent matrices on the right.
At the end, multiplication by $b_0$ selects the first entry of the final
row. The factor $D$ in the output corrects for this one-coordinate
projection, as explained in Section~\ref{sec:estimator}.

\subsection{The incoming-arc scan}
Suppose all states of degree $p-1$ have been computed. Before inserting
the next vertex, form an incoming table
\begin{equation}\label{eq:app-incoming}
 U_{p,v}[S]=\sum_{u:(u,v)\in E(G)}X_{p-1,u}[S],
 \qquad |S|=p-1.
\end{equation}
Its exterior degree is still $p-1$. The subscript $p$ says which update
will use it; no new exterior vector has yet been inserted.

To compute this table, initialize it to zero and scan the adjacency
lists. An arc $(u,v)$ contributes one row addition
$U_{p,v}[S]\mathrel{+}=X_{p-1,u}[S]$ for each $(p-1)$-set $S$.
A row has $D$ entries, so this scan uses
$\OO(mD\binom{k}{p-1})$ arithmetic operations. It does not examine
pairs of incoming walks or pairs of exterior rows.

Why can we add the incoming rows before applying the matrix? All arcs
entering $v$ at this layer use the same $M_{p,v}$ and the same $g_v$.
Bilinearity gives
\[
 \sum_{u:(u,v)\in E}(X_{p-1,u}\wedge g_v)M_{p,v}
 =\left(\left(\sum_{u:(u,v)\in E}X_{p-1,u}\right)\wedge g_v\right)M_{p,v}.
\]
Consequently the matrix multiplication is performed once per endpoint
and exterior row, after the arc scan. This is why its cost below has a
factor $n$, whereas the incoming scan has a factor $m$.

\subsection{Appending one vector: every coordinate and every sign}
Fix $v$ and write $U=U_{p,v}$. Expand
$U=\sum_{|S|=p-1}e_S\otimes U[S]$ and
$g_v=\sum_{j=1}^k g_v(j)e_j$. For a term indexed by $(S,j)$, there
are two possibilities. If $j\in S$, then $e_S\wedge e_j=0$. Otherwise
put $T=S\cup\{j\}$. The new $e_j$ starts at the right end of $e_S$;
to obtain the increasing order on $T$, move it past the elements of $S$
that are larger than $j$.

If $j$ occupies position $q$ in the increasing order on $T$, there are
$p-q$ such elements. Therefore the intermediate row after the wedge is
\begin{equation}\label{eq:app-wedge}
 H_{p,v}[T]=\sum_{j\in T}(-1)^{p-\operatorname{pos}_T(j)}
              g_v(j)U_{p,v}[T\setminus\{j\}],\qquad |T|=p.
\end{equation}
The sign comes from moving the appended coordinate to its position;
it does not depend on the graph or on the auxiliary coordinate.

For $k=3$, the entire degree-one to degree-two update is
\begin{align*}
 H[\{1,2\}]&=g_v(2)U[\{1\}]-g_v(1)U[\{2\}],\\
 H[\{1,3\}]&=g_v(3)U[\{1\}]-g_v(1)U[\{3\}],\\
 H[\{2,3\}]&=g_v(3)U[\{2\}]-g_v(2)U[\{3\}].
\end{align*}
The degree-two to degree-three update has just one output row:
\[
 H[\{1,2,3\}]=g_v(3)U[\{1,2\}]
               -g_v(2)U[\{1,3\}]+g_v(1)U[\{2,3\}].
\]
These are row-vector equations: multiply and add each of the $D$
auxiliary entries. There is no multiplication between two rows of $U$.

\paragraph{Counting the incidences.}
For each of the $\binom{k}{p}$ output sets $T$, the displayed sum has
$p$ terms. Thus there are exactly $p\binom{k}{p}$ possible incidences.
Equivalently, each input set $S$ has $k-p+1$ available new coordinates,
giving $(k-p+1)\binom{k}{p-1}$ incidences. These count the same pairs:
\begin{equation}\label{eq:app-incidences}
 p\binom{k}{p}=(k-p+1)\binom{k}{p-1}.
\end{equation}
Some numerical entries may be zero, but the algorithm can simply process
all these incidences. A scalar times a row, accumulated into another row,
costs $\OO(D)$ operations. The wedge updates at all endpoints therefore
cost $\OO(nDp\binom{k}{p})$ at this layer.

\subsection{Multiplying the auxiliary row}
After the wedge update, set $X_{p,v}[T]=H_{p,v}[T]M_{p,v}$. At the
level of scalar entries this means
\begin{equation}\label{eq:app-matrix}
 X_{p,v}[T]_t=\sum_{a=1}^D H_{p,v}[T]_a(M_{p,v})_{a,t},
 \qquad t\in[D].
\end{equation}
One output entry takes $D$ multiplications and at most $D-1$ additions;
there are $D$ output entries. Hence one row--matrix multiplication costs
$\OO(D^2)$, and all such multiplications at this layer cost
$\OO(nD^2\binom{k}{p})$.

The full state is a table with $\binom{k}{p}$ rows of length $D$.
The matrix acts separately on those rows; it never multiplies the table
by another exterior table. We can form one temporary row $H_{p,v}[T]$,
transform it, and store the result, then reuse the temporary row for the
next $T$. A full table of $H$ is unnecessary.

\subsection{A complete numerical example}
\label{app:path-worked-example}
Take the directed diamond with arcs $ab,bc,ad,dc$. Its two three-vertex
paths are $(a,b,c)$ and $(a,d,c)$. To see the arithmetic explicitly, set
$k=3$ and $D=2$ and fix the following vector values:
\[
 g_a=(1,0,1)^\top,\quad g_b=(0,1,1)^\top,\quad
 g_d=(1,1,0)^\top,\quad g_c=(1,2,4)^\top.
\]
Take every matrix to be $I_2$ except
\[
 M_{2,b}=\begin{pmatrix}1&2\\0&1\end{pmatrix},\qquad
 M_{2,d}=\begin{pmatrix}2&-1\\1&0\end{pmatrix},\qquad
 M_{3,c}=\begin{pmatrix}1&1\\1&2\end{pmatrix}.
\]
These are fixed inputs for checking the recurrence, not a proposed random
sampling rule. The algebraic expansion is valid for every such assignment.
The choice $D=2$ keeps each row visible; the approximation algorithm uses
$D=k$. In the tables below we abbreviate $\{1,2\}$ by $12$, and similarly
for the other sets. Every parenthesized pair is an auxiliary row.

\paragraph{Layer one.}
Since $b_0^\top M_{1,v}=(1,0)$ for every $v$, initialization multiplies
this row by each coordinate of $g_v$:
\begin{center}
\begin{tabular}{@{}c r r r@{}}
\toprule
Endpoint $v$ & $X_{1,v}[1]$ & $X_{1,v}[2]$ & $X_{1,v}[3]$\\
\midrule
$a$ & $(1,0)$ & $(0,0)$ & $(1,0)$\\
$b$ & $(0,0)$ & $(1,0)$ & $(1,0)$\\
$d$ & $(1,0)$ & $(1,0)$ & $(0,0)$\\
$c$ & $(1,0)$ & $(2,0)$ & $(4,0)$\\
\bottomrule
\end{tabular}
\end{center}

\paragraph{The first arc scan.}
Vertices $b$ and $d$ each receive $X_{1,a}$, vertex $c$ receives
$X_{1,b}+X_{1,d}$, and $a$ receives nothing. Thus
\begin{center}
\begin{tabular}{@{}c r r r@{}}
\toprule
Endpoint $v$ & $U_{2,v}[1]$ & $U_{2,v}[2]$ & $U_{2,v}[3]$\\
\midrule
$a$ & $(0,0)$ & $(0,0)$ & $(0,0)$\\
$b$ & $(1,0)$ & $(0,0)$ & $(1,0)$\\
$d$ & $(1,0)$ & $(0,0)$ & $(1,0)$\\
$c$ & $(1,0)$ & $(2,0)$ & $(1,0)$\\
\bottomrule
\end{tabular}
\end{center}

\paragraph{The first wedge and matrix updates.}
At $b$, the three wedge rows are $(1,0)$, $(1,0)$, and $(-1,0)$.
Multiplication by $M_{2,b}$ turns them into $(1,2)$, $(1,2)$, and
$(-1,-2)$. At $d$, the wedge rows are $(1,0)$, $(-1,0)$, and
$(-1,0)$, and multiplication by $M_{2,d}$ gives $(2,-1)$,
$(-2,1)$, and $(-2,1)$. At $c$, the wedge rows are
$(0,0)$, $(3,0)$, and $(6,0)$, and $M_{2,c}=I_2$. The complete table is
\begin{center}
\begin{tabular}{@{}c r r r@{}}
\toprule
Endpoint $v$ & $X_{2,v}[12]$ & $X_{2,v}[13]$ & $X_{2,v}[23]$\\
\midrule
$a$ & $(0,0)$ & $(0,0)$ & $(0,0)$\\
$b$ & $(1,2)$ & $(1,2)$ & $(-1,-2)$\\
$d$ & $(2,-1)$ & $(-2,1)$ & $(-2,1)$\\
$c$ & $(0,0)$ & $(3,0)$ & $(6,0)$\\
\bottomrule
\end{tabular}
\end{center}
For instance, $X_{2,c}$ contains the two-vertex paths $(b,c)$ and
$(d,c)$. Although neither can be extended in this graph, it is correct
and harmless to compute their rows.

\paragraph{Layer three and the output.}
Only $c$ receives a nonzero degree-two incoming table. Adding the rows
at $b$ and $d$ gives
$U_{3,c}[12]=(3,1)$, $U_{3,c}[13]=(-1,3)$, and
$U_{3,c}[23]=(-3,-1)$. Consequently
\[
 H_{3,c}[123]=4(3,1)-2(-1,3)+(-3,-1)=(11,-3),
 \qquad X_{3,c}[123]=(11,-3)M_{3,c}=(8,5).
\]
All other degree-three states are zero. The endpoint sum is therefore
$z=(8,5)$, its first coordinate is $zb_0=8$, and the trial output is
$Y=2\cdot8^2/3!=64/3$.

We can check this result directly against the two path contributions.
Their determinants are $\det(g_a,g_b,g_c)=1$ and
$\det(g_a,g_d,g_c)=5$. Their auxiliary rows are
$(1,0)M_{2,b}M_{3,c}=(3,5)$ and
$(1,0)M_{2,d}M_{3,c}=(1,0)$. Their determinant-weighted sum is
$(3,5)+5(1,0)=(8,5)$, exactly the computed row.
The numerical output of this fixed assignment need not be the count two.
The count is recovered in expectation under the prescribed random law,
and concentration is obtained by independent trials.

\paragraph{Where a repeated vertex vanishes.}
If the arc $ba$ were also present, the walk $(a,b,a)$ would occur in
the walk expansion. Its degree-two minors are $1,1,-1$, so appending
$g_a=(1,0,1)^\top$ gives $1\cdot1-0\cdot1+1\cdot(-1)=0$ in the
top coordinate. Every auxiliary entry is multiplied by this same zero.
The new matrix $M_{3,a}$ does not revive the term. This illustrates how
repetition is removed by arithmetic within the coordinate update,
without a membership test on visited host vertices.

\subsection{Complete pseudocode and the table invariant}
\label{app:path-pseudocode}
Here is a direct implementation of one trial. Each assignment of a row
means an assignment of all its $D$ entries. The matrices may be sampled
at the beginning, as written, or generated when first needed as described
in Appendix~\ref{app:path-storage}.

\begin{enumerate}[leftmargin=1.7em,itemsep=5pt]
\item \textbf{Prepare the trial.} Set $b_0=e_1$. Sample all $g_v$ and
      $M_{p,v}$ independently from the laws above. Prepare an indexing
      scheme for the subsets of $[k]$.
\item \textbf{Initialize the previous layer.} For every endpoint $v$,
      obtain the first row $r=b_0^\top M_{1,v}$. For each $j\in[k]$,
      set $\mathrm{Prev}[v,\{j\}]=g_v(j)r$.
\item \textbf{For each layer $p=2,\ldots,k$, execute all of the following.}
      \begin{enumerate}[label=(\alph*),leftmargin=1.8em,itemsep=4pt]
      \item Set $\mathrm{In}[v,S]=0$ for every $v$ and $(p-1)$-set $S$.
      \item For each arc $(u,v)$, and each $(p-1)$-set $S$, add
            $\mathrm{Prev}[u,S]$ to $\mathrm{In}[v,S]$.
      \item Allocate $\mathrm{Next}$ with rows indexed by $(v,T)$,
            $|T|=p$. For each endpoint $v$ and each such $T$, initialize
            one temporary row $h=0$. For $j$ in the increasing order on
            $T$, accumulate
            \[
             h\mathrel{+}=(-1)^{p-\operatorname{pos}_T(j)}
                            g_v(j)\mathrm{In}[v,T\setminus\{j\}].
            \]
            Then set $\mathrm{Next}[v,T]=hM_{p,v}$, computing each
            output entry by~\eqref{eq:app-matrix}.
      \item Once all endpoints have been processed, release
            $\mathrm{Prev}$ and $\mathrm{In}$, and rename
            $\mathrm{Next}$ as $\mathrm{Prev}$.
      \end{enumerate}
\item \textbf{Finish the trial.} Set $z=\sum_v\mathrm{Prev}[v,[k]]$,
      read its first coordinate $z_1$, and return $Y=Dz_1^2/k!$.
\end{enumerate}
The layer loop includes both the arc scan and all endpoint updates. In
particular, no degree-$p$ value is used as a degree-$(p-1)$ input. This
separation also makes the algorithm valid when the host graph has directed
cycles.

\paragraph{Invariant.}
After initialization, and after each completed layer $p$, the row
$\mathrm{Prev}[v,S]$ equals
\begin{equation}\label{eq:app-coordinate-invariant}
 \sum_{\substack{P=(v_1,\ldots,v_p)\text{ a directed walk}\\v_p=v}}
       \det(G_{P,S})\,b_0^\top M_{1,v_1}\cdots M_{p,v_p},\qquad |S|=p.
\end{equation}
For $p=1$, a one-column minor is $g_v(j)$, so this is the initialization.
Assume it holds at degree $p-1$. The arc scan lists every possible
extension exactly once, through its last arc. Formula~\eqref{eq:app-wedge}
expands the resulting minor along its last column. Its cofactor sign
$(-1)^{q+p}$ equals $(-1)^{p-q}$, the sign used in the update.
Formula~\eqref{eq:app-matrix} appends the correct matrix to the row.
Distributivity proves the invariant at degree $p$.

If a walk repeats a vertex, two columns of every relevant minor are
equal, so its contribution is zero. At degree $k$, the unique exterior
coordinate is $[k]$, and its minors are the determinants $\Delta_P$
from the main text. Thus the endpoint sum and final projection produce
exactly $Z=\sum_{P\in\calP_k}\Delta_P\chi(P)$ and $Y=Z^2/k!$.
This proves the correctness of the coordinate procedure for every fixed
assignment, before any probability calculation is used.

\subsection{Operation counts, layer by layer}
\label{app:path-costs}
We count scalar additions and multiplications using the elementary dense
row--matrix algorithm. Put $B_p=\binom{k}{p}$. A row addition or a
scalar-times-row accumulation costs $\OO(D)$; a row--matrix product costs
$\OO(D^2)$. For one layer $p\ge2$, the table gives every numerical step:
\begin{center}
\small
\begin{tabular}{@{}p{0.42\linewidth}p{0.25\linewidth}p{0.23\linewidth}@{}}
\toprule
Step & Number of items & Total operations\\
\midrule
Clear incoming rows & $nB_{p-1}$ rows & $\OO(nDB_{p-1})$\\
Scan arcs and add rows & $mB_{p-1}$ rows & $\OO(mDB_{p-1})$\\
Accumulate wedge terms & $npB_p$ incidences & $\OO(nDpB_p)$\\
Multiply auxiliary rows & $nB_p$ products & $\OO(nD^2B_p)$\\
Write next-layer rows & $nB_p$ rows & $\OO(nDB_p)$\\
\bottomrule
\end{tabular}
\end{center}
No factor in this table is the number of paths. The entries depend only
on the graph size, the exterior degree, and the auxiliary dimension.

Initialization costs $\OO(nkD)$ operations: read one matrix row at each
vertex and scale it for each of the $k$ singleton coordinates. Sampling
all exterior vectors and matrices requires $nk+knD^2$ scalar draws.
In the ideal Gaussian model we account for each draw at unit cost;
Section~\ref{sec:finite-new} supplies the finite implementation. Summing
endpoints at the end costs $\OO(nD)$, and $k!$ can be precomputed in
$\OO(k)$ multiplications.

\paragraph{Why the sum of coordinate counts is $2^k$.}
Every subset of $[k]$ is obtained by deciding, independently for each
coordinate, whether it is present. There are $2^k$ choices in total.
Partitioning them by their cardinality gives
\begin{equation}\label{eq:app-binomial-sum}
 \sum_{p=0}^k B_p=2^k.
\end{equation}
In particular, the arc scans sum $B_1+\cdots+B_{k-1}=2^k-2$.
Although a different table shape is used at each layer, we sum these
shapes; we do not multiply their sizes.

\paragraph{Why the extra wedge factor is only $k$.}
Count pairs $(T,j)$ with $T\subseteq[k]$ and $j\in T$. Counting by
$|T|=p$ gives $\sum_p pB_p$. Alternatively, first choose $j$ in $k$
ways, and then choose an arbitrary subset of the other $k-1$ coordinates.
Therefore
\begin{equation}\label{eq:app-marked-binomial-sum}
 \sum_{p=0}^k pB_p=k2^{k-1}.
\end{equation}
The degree-one insertions account for $k$ of these incidences; the wedge
updates at layers $2,\ldots,k$ account for the remaining $k2^{k-1}-k$.
This is the full count of possible exterior incidences over a trial.

\paragraph{Adding all operations.}
The arc scans cost $\OO(mD2^k)$, the wedge steps cost
$\OO(nDk2^{k-1})$, and the matrix products cost $\OO(nD^22^k)$.
Clearing and writing the tables costs $\OO(nD2^k)$, which is absorbed
by the preceding terms for $D\ge1$. Including initialization and random
generation gives exactly the bound used in the main text:
\begin{equation}\label{eq:app-total-cost}
 \OO\bigl(2^k(mD+nkD+nD^2)+knD^2\bigr).
\end{equation}
With $D=k$, this is $\OO(2^k(km+k^2n))$: the generation term $nk^3$
is absorbed because $k\le2^k$. Thus one trial uses
$2^k\operatorname{poly}(k)(n+m)$ arithmetic operations.

\paragraph{Subset indexing does not introduce a second exponential factor.}
A concrete implementation represents $S$ by its $k$-bit mask and stores
its row at a dense address within its degree. One can build lists of
masks by degree and an array that maps each mask to this dense address.
There are $2^k$ masks in all; building the lists and addresses by inspecting
their bits takes $\OO(k2^k)$ work and $\OO(2^k)$ index words. Given
$T$ and $j\in T$, clear its $j$-th bit and look up the address of
$T\setminus\{j\}$. No search through the previous table is needed.

With $k$-bit index words, these address lookups take constant time in the
usual random-access model. Even enumerating the elements of each $T$
by a $k$-position scan adds only $\OO(nk2^k)$ work over all layers,
which is absorbed by~\eqref{eq:app-total-cost}. Index preparation is
also absorbed since $n\ge k$ and $D\ge1$, and can be reused across
trials. Charging operations on these index words at their bit cost adds
only a polynomial factor in $k$ and the logarithm of the table size.

\subsection{Where larger exponential costs would arise}
The number $2^k$ by itself is not enough to prove the running time:
we must also inspect the work needed to update the table. For two
arbitrary exterior elements, a direct product loops over disjoint pairs
$(A,B)$. Each coordinate lies in $A$, in $B$, or in neither, giving
$3^k$ possible pairs across all degrees. A path transition has a special
form: its second exterior factor is a single vector. Only pairs
$(S,\{j\})$ occur. Their count is the $k2^{k-1}$ derived above.

Similarly, explicitly storing pairs of degree-$p$ exterior coordinates
would require $B_p^2$ positions at that degree. Summing gives
\[
 \sum_{p=0}^k\binom{k}{p}^2=\binom{2k}{k}.
\]
One proof counts $k$-subsets of two disjoint $k$-sets by the number chosen
from the first. Our trial never stores this paired table. It propagates
one row per exterior coordinate and squares only the final scalar.
The probabilistic moment argument justifies using independent repetitions
to estimate the required mass. These comparisons describe concrete table
operations; they are not lower bounds for other possible algorithms.

\subsection{Space usage and reuse between layers}
\label{app:path-storage}
At layer $p$, the pseudocode simultaneously stores a previous table and
an incoming table of size $nDB_{p-1}$ each, and a next table of size
$nDB_p$. Only one temporary row of length $D$ is needed for the wedge
and matrix step. Therefore the numerical tables use at most
$nD(2B_{p-1}+B_p)+\OO(D)$ scalar words. Writing
$B_{\max}=\max_{0\le p\le k}B_p=\binom{k}{\lfloor k/2\rfloor}$,
this is $\OO(nDB_{\max})$.

Retaining the exterior vectors takes $nk$ words. If all matrices are
sampled and retained at the beginning, they occupy $knD^2$ words. The
subset-address lists require $\OO(2^k)$ index words; they fit the table
bound because $2^k\le(k+1)B_{\max}$ and $k\le n$. Thus, apart from
the input graph, the total working storage is
\[
 \OO(nDB_{\max}+nk+knD^2)
\]
words. This is the storage convention in
Proposition~\ref{prop:implementation-cost}.

For paths alone, we can generate $M_{p,v}$ when processing endpoint $v$
at layer $p$, use it on every exterior row there, and then discard it.
At initialization, similarly generate $M_{1,v}$ before producing all
singleton rows at $v$. Fresh independent draws give the same joint law
as sampling every matrix in advance. This reduces matrix storage to
$\OO(D^2)$ words, so the working space becomes
$\OO(nDB_{\max}+nk+D^2)$. The vectors $g_v$ remain stored and reused
at every layer. Tree separator computations reuse role--image matrices
across many assignments, so the path streaming observation should not
be substituted into that algorithm without accounting for its different
reuse requirements.

\subsection{From one trial to a relative approximation}
Take $D=k$ and $C_k=3e^2(k+1)(k+2)/2$. By
Theorem~\ref{thm:human-path-estimator}, one trial has expectation $N_k$
and second moment at most $C_kN_k^2$. Form each group average from
$R=\lceil8C_k\varepsilon^{-2}\rceil$ independent trials. For $N_k>0$,
Chebyshev's inequality bounds the probability of a relative error larger
than $\varepsilon$ in this average by $1/8$.

Take $q$ to be the smallest positive odd integer at least
$2\log_2(1/\delta)$. Form $q$ independent groups and return the median
of their averages. If the median fails, at least half
the groups fail. The same subset bound as in the main text gives failure
probability at most $2^q(1/8)^{q/2}=2^{-q/2}\le\delta$.
If $N_k=0$, every trial is identically zero. Hence the total number of
trials is $\OO(k^2\varepsilon^{-2}\log(2/\delta))$.

Multiplying this by~\eqref{eq:app-total-cost} gives, for the elementary
implementation with $D=k$, the explicit arithmetic bound
\[
 \OO\bigl(2^k(k^3m+k^4n)\varepsilon^{-2}\log(2/\delta)\bigr).
\]
In particular, the repetition count introduces only polynomial factors;
it does not change the exponential base. Trials can run sequentially,
reusing the same table allocation. Keep one running sum for the current
group and the $q$ completed group averages. Beyond the working space of
one trial, this requires only $\OO(q)$ numerical words. All random
vectors and matrices are freshly resampled for each trial.




%
%
\end{document}